\documentclass[11pt,report]{article}
\usepackage{amsfonts}
\usepackage{amssymb}
\usepackage{amsbsy}
\usepackage{amsmath}
\usepackage{bbm}
\newtheorem{definition}{Definition}
\newtheorem{lemma}[definition]{Lemma}
\newtheorem{theorem}[definition]{Theorem}
\newtheorem{remark}[definition]{Remark}
\newtheorem{proposition}[definition]{Proposition}
\newtheorem{corollary}[definition]{Corollary}

\newenvironment{proof}
  {\par \noindent {\bf Proof\ }}
  {\hfill {\small $\blacksquare$} \par \medskip}
\newcommand{\mR}{\mathbb{R}}
\newcommand{\mC}{\mathbb{C}}
\newcommand{\mN}{\mathbb{N}}
\newcommand{\mZ}{\mathbb{Z}}
\newcommand{\mM}{\mathbb{M}}
\newcommand{\rH}{\mathrm{H}}
\newcommand{\cB}{\mathcal{B}}
\newcommand{\cC}{\mathcal{C}}
\newcommand{\cD}{\mathcal{D}}
\newcommand{\cF}{\mathcal{F}}
\newcommand{\cH}{\mathcal{H}}
\newcommand{\cL}{\mathcal{L}}
\newcommand{\cX}{\mathcal{X}}

\newcommand{\cO}{\mathcal{O}}
\newcommand{\cU}{\mathcal{U}}

\newcommand{\cW}{\mathcal{W}}
\newcommand{\cT}{\mathcal{T}}
\newcommand{\bV}{\mathbf{V}}
\newcommand{\eps}{\epsilon}
\newcommand{\ol}{\overline}
\newcommand{\KP}{\mathrm{kp}}
\newcommand{\EM}{\mathrm{em}}
\newcommand{\per}{\mathrm{per}}
\newcommand{\e}{\mathrm{e}}

\newcommand{\abs}[1]{{\left\vert {#1} \right\vert}}
\newcommand{\norma}[1]{{\| {#1} \|}}
\newcommand{\bk}[1]{{\left\langle #1 \right\rangle}}
\newcommand{\cara}{\mathbbm{1}}
\newcommand{\carB}{\mathbbm{1}_{\cB/\eps}}
\newcommand{\gini}{g^\mathit{in}}
\newcommand{\hini}{h^\mathit{in}}

\newcommand{\psini}{\psi^{\mathit{in},\eps}}
\newcommand{\ginieps}{g^{\mathit{in},\eps}}
\DeclareMathOperator{\DIV}{div}
\DeclareMathOperator{\supp}{supp}
\DeclareMathOperator{\Span}{span}
\numberwithin{equation}{section}
\numberwithin{definition}{section}
\newcommand{\ds}{\displaystyle}

\begin{document}
\title{Quantum Transport in Crystals: \\ Effective Mass Theorem  and K$\cdot$P Hamiltonians}
%
%
\author{Luigi Barletti$^{1}$  and  Naoufel Ben Abdallah$^{2}$}

%
\date{}
%
%
\maketitle
\begin{center}
$^{1}$Dipartimento di Matematica, 
Universit\`a di Firenze,
Viale Morgagni 67/A, 50134 Firenze, Italy,
{\em barletti@math.unifi.it}\\
$^{2}$Institut de Math\'ematiques de Toulouse, 
Universit\'e de Toulouse
Univ.~Paul Sabatier,
118 route de Narbonne, 31062 Toulouse, France,
{\em naoufel@math.univ-toulouse.fr}   
\end{center}
\begin{abstract}
In this paper the effective mass approximation and k$\cdot$p multi-band models,
describing quantum evolution of electrons in a crystal lattice, are discussed. 
Electrons are assumed to move in both a periodic potential and a macroscopic one.
The typical period $\eps$ of the periodic potential is assumed to be very
small, while the macroscopic potential acts on a much bigger length scale. 
Such homogenization asymptotic is investigated by using the envelope-function 
decomposition of the electron wave function. 
If the external potential is smooth enough, the k$\cdot$p and effective mass models, 
well known in solid-state physics, are proved to be close (in strong sense) to the 
exact dynamics. 
Moreover, the position density of the electrons is proved to converge weakly
to its effective mass approximation.
\end{abstract}
\section{Introduction}
\label{Intro}
The effective mass approximation is a common approximation in solid state physics \cite{Bastard,Ashcroft76,Wenckebach99} 
and states roughly speaking that the motion of electrons in a periodic potential can be replaced with a good approximation by 
the motion of a fictitious particle in vacuum but with a modified mass called the effective mass of the electron. 
This approximation is valid when the lattice period is small compared to the observation length scale, 
it relies on the Bloch decomposition theorem for the Schr\"odinger equation with a periodic potential. 
The effective mass is actually a tensor and  depends on the energy band in which the electron ``live's'.  
One of the most important references in the Physics literature on the subject is the paper of Kohn and Luttinger 
\cite{LuttingerKohn55} which dates back to 1955. 
As for rigorous mathematical treatment of this problem, we are aware of the work of Poupaud and Ringhofer  
\cite{PoupaudRinghofer96} and that of Allaire and Piatnitski \cite{Allaire05}.  
The aim of the present work is to provide an alternative mathematical treatment which is based on the original 
work of Kohn and Luttinger. 
Like in \cite{Allaire05} (see also \cite{Allaire04} and \cite{Allaire06} for related problems), we consider the  scaled 
Schr\"odinger equation
$$
  i\partial_t\, \psi(t,x) = \left( -\frac{1}{2}\Delta + \frac{1}{\eps^2}
   \,W_\cL\left(\frac{x}{\eps}\right) + V\left(x,\frac{x}{\eps}\right) \right) \psi(t,x),
$$
where $W_\cL(z)$ is a periodic potential with the periodicity of a lattice $\cL$, representing the 
crystal ions, while $V(x,z)$ represents an external potential.
The latter is assumed to act both on the macroscopic scale $x$ and on the microscopic scale $z = x/\eps$,
and to be $\cL$-periodic with respect to $z$.
The small parameter $\eps$ is interpreted as the so-called ``lattice constant'', that is the 
typical separation between lattice sites. 
Note that the scaling of the Schr\"odinger equation \eqref{SE1} is a 
homogenization scaling \cite{Allaire05,PoupaudRinghofer96}.
As mentioned above, the analysis of the limit $\eps \to 0$ has been done in Refs.~\cite{Allaire05} and 
\cite{PoupaudRinghofer96} by different techniques. 
In \cite{PoupaudRinghofer96}, the analysis is done indirectly by means of Wigner functions techniques.
Using Bloch functions which diagonalize the periodic Hamiltonian, a Wigner function is constructed.
The limit $\eps \to 0$ is done in the Wigner equation and is reinterpreted as the Wigner transform
of an effective mass Schr\"odinger equation.
In \cite{Allaire05}, the problem is tackled differently  thanks to homogenization techniques, mainly
double-scale limits.  
The wave function is spanned on the Bloch basis and
the limiting equation is obtained by expanding around zero wavevector
the Bloch functions and the energy bands.
 \par
The approach we adopt in this paper is completely different from \cite{PoupaudRinghofer96}
and somehow related to \cite{Allaire05} although the techniques are different.
The main idea, borrowed from the celebrated work of Kohn and Luttinger \cite{LuttingerKohn55},
consists of expanding the wave function on a modified Bloch basis.
This choice of basis does not allow to completely diagonalize the periodic part of the Hamiltonian,
but completely separates the ``oscillating'' part of the wave function from its slowly varying one.
By doing so, we introduce a so-called envelope function decomposition of the wave function and
rewrite the Schr\"odinger equation as an infinite system of coupled Schr\"odinger equations.
Each of the envelope functions has a fast oscillating scale in time with a frequency related to the 
energy band for vanishing wavevector. 
Therefore adiabatic decoupling occurs as it is commonly the 
case for fast oscillating systems \cite{hagedorn-joye,panati,spohn-teufel,teufel}. 
The action of the macroscopic potential becomes in the envelope function formulation a convolution 
operator in both the position variable  and band index. 
The limit of this operator becomes a multiplication operator in position by a matrix potential 
(in the band index).
The analysis of this limiting process is obtained through simple Fourier-like  analysis and 
perturbation of point spectra of self-adjoint operators.
The method allows to handle an infinite number of Bloch waves and also derive the so-called 
k$\cdot$p Hamiltonian as an intermediate model between the original Schr\"odinger equation 
and its limiting effective mass approximation.
\par
\medskip
The outline of the paper is as follows. Section \ref{sec2} is devoted to the presentation of the functional setting, 
notations as well as the main result of the paper. 
As mentioned above, the Schr\"odinger equation is reformulated as an infinite system of coupled Schr\"odinger 
equations, where the coupling comes both from the differential part and from the potential part. 
In Section \ref{sec3}, we concentrate on the potential part and analyze its limit. Section \ref{sec4} is devoted to 
the diagonalization of the differential part  and to the expansion of the corresponding eigenvalues in the Fourier space. 
In Section  \ref{sec5}, we analyze the convergence of the solution of the Schr\"odinger equation towards its 
effective mass approximation. 
The method relies on the definition of intermediate models and the comparison of their respective dynamics. 
Some comments are done in Section \ref{sec6} while some proofs are postponed to Section \ref{post}.
\section{Notations and main results}
\label{sec2}
\subsection{Bloch decomposition}
\label{SecDEF}
Let us consider the operator
\begin{equation}
\label{ScaledHamiltonian}
   H_\cL^\eps = -\textstyle{\frac{1}{2}} \Delta + \frac{1}{\eps^2}\,W_\cL\Big(\frac{x}{\eps}\Big),
\end{equation}
where $W_\cL$ is a bounded $\cL$-periodic potential where the lattice $\cL$ is defined by
\begin{equation}
\label{DirLatt}
  \cL = \left\{ Lz \ \big| \   z \in \mZ^d \right\}  \subset \mR^d,
\end{equation}
where $L$ be a $d\times d$ matrix with $\det L \not= 0$.
The centered fundamental domain $\cC$ of $\cL$ is, by definition,
\begin{equation}
\label{Cell}
  \cC = \left\{ Lt \ \Big| \  t \in \Big[-{1\over 2},{1\over 2}\Big]^d \right\}.
\end{equation}
Note that the volume measure $\abs{\cC}$ of $\cC$ is given by $\abs{\cC} = \abs{\det L}$.
The {\em reciprocal lattice} $\cL^*$ is, by definition, the lattice generated by the matrix $L^*$
such that
\begin{equation}
\label{RecLatt}
  L^T L^* = 2\pi I.
\end{equation}
The {\em Brillouin zone} $\cB$ is the centered fundamental domain 
of $\cL^*$, i.e.\footnote{In solid state physics the Brillouin zone used has a slightly 
different definition. However, the two definitions are equivalent to our purposes.}
\begin{equation}
\label{Brillo}
  \cB = \left\{ L^*t \ \Big| \  t \in \Big[-{1\over 2} , {1\over 2}\Big]^d \right\} .
\end{equation}
Thus, we clearly have 
\begin{equation}
\label{CB}
   \abs{\cC}\,\abs{\cB} = (2\pi)^d.
\end{equation}
We assume without loss of generality that the periodic potential is  larger than one ($W_\cL \geq 1$).  
In solid state physics, $W_\cL$ is interpreted as the electrostatic potential generated by the ions 
of the  crystal lattice \cite{Ashcroft76}.  
With the change of variables $z = x/\eps$, the operator $ H_\cL^\eps$ turns to ${1\over \eps^2}  H_\cL^1$, 
where $H_\cL^1$ is given by \eqref{ScaledHamiltonian} with $\eps = 1$. 
This operator has a band structure which is given by the celebrated Bloch theorem \cite{ReedSimonIV78}.
\begin{definition}
\label{BlochDef}
For any $k\in \cB$, the fiber Hamiltonian
\begin{equation}
\label{HkDef}
  H_\cL(k) = \frac{1}{2} \abs{k}^2 - ik\cdot\nabla -\textstyle{\frac{1}{2}} \Delta + W_\cL.
\end{equation}
defined on $L^2(\cC)$ with periodic boundary condition has a compact resolvent. 
Its eigenfunctions form an orthonormal sequence of  periodic solutions $(u_{n,k})_{n\in\mN})$ 
solving the eigenvalue problem 
\begin{equation}
\label{Bloch}
  H_\cL(k)u_{n,k} = E_n(k)u_{n,k}
\end{equation}
The functions $u_{n,k}$ are the so-called {\em Bloch functions} and the eigenvalues $E_n(k)$ are 
the {\em energy bands} of the crystal.
For each fixed value of $k\in\cB$, the set $\{u_{n,k} \mid n \in \mN\}$ is a Hilbert 
basis of $L^2(\cC)$ \cite{BerezinShubin91,ReedSimonIV78}. The Bloch waves defined for $k\in\cB$ and $n\in \mN$ by 
$$
 \cX_{n,k}^\textsc{b}(x) = \abs{\cB}^{-1/2} \,\cara_\cB(k)\, \e^{ik\cdot x}\,u_{n,k}(x)
$$
form a complete basis of $L^2(\mR^d)$ and satisfy the equation
$$
  H_\cL^1\cX_{n,k}^\textsc{b} = E_n(k)\cX_{n,k}^\textsc{b}.
$$
The scaled Bloch functions are given by
$$
 \cX_{n,k}^{\textsc{b},\eps}(x) = \abs{\cB}^{-1/2} \,\carB(k)\, \e^{ik\cdot x}\,u_{n,\eps k}(x)
$$
and they satisfy
$$
  H_\cL^\eps\cX_{n,k}^{\textsc{b},\eps}  = \frac{E_n(\eps k)}{\eps^2}\cX_{n,k}^{\textsc{b},\eps}.
$$
\end{definition}
In order to analyze the limit $\eps\to 0$, the usual starting point is to decompose the wave function on the Bloch 
wave functions. This decomposition was in particular used in \cite{Allaire05}. 
This has the big advantage of completely diagonalizing the periodic Hamiltonian, but since the wave vector
appears both in the plane wave $\e^{i k\cdot x}$ and in the standing periodic function $u_{n,\eps k}$, the separation 
between the fast oscillating scale and the slow motion carried by the plane wave is not immediate. 
We follow in this work the idea of Kohn and Luttinger  \cite{LuttingerKohn55} who decompose the wave function 
on the basis 
\begin{equation}
\label{kl}
 \cX_{n,k}^\textsc{lk}(x) = \abs{\cB}^{-1/2} \,\cara_\cB(k)\, \e^{ik\cdot x}\,u_{n,0}(x).
\end{equation}
The family $\cX_{n,k}^\textsc{lk}$ is also a complete orthonormal basis of $L^2(\mR^3)$ but only partially 
diagonalizes $H^1_\cL$ since 
\begin{equation}
\label{HLK}
\begin{aligned}
  H_\cL^1\cX_{n,k}^\textsc{lk} &= \abs{\cB}^{-1/2} \carB(k)\, \e^{ik\cdot x}
  \left[\frac{1}{2} \abs{k}^2 - ik\cdot\nabla + E_n(0) \right]u_{n,0}
\\
 &= \abs{\cB}^{-1/2} \carB(k)\, \e^{ik\cdot x} \sum_{n'} 
   \left[\frac{1}{2} \abs{k}^2\delta_{nn'} - ik\cdot P_{nn'} 
   + E_n\delta_{nn'}\right] u_{n',0}
\\
 &= \sum_{n'} \left[\frac{1}{2} \abs{k}^2\delta_{nn'} - ik\cdot P_{nn'} 
   + E_n\delta_{nn'}\right]\cX_{n',k}^\textsc{lk}\,.
\end{aligned}
\end{equation}
Here, $E_n = E_n(0)$ and
\begin{equation}
\label{Pdef}
  P_{nn'} = \int_\cC \ol u_{n,0}(x) \nabla u_{n',0}(x)\,dx 
\end{equation}
are the matrix elements of the gradient operator between Bloch functions.
The interest of the Luttinger-Kohn wave functions is that the wave vector $k$ only appears in the plane wave and not in the standing periodic part $u_{n,0}$.  This will allow us to decompose the wave function in a nice way 
for which we will prove some Hilbert analysis type results. This is the envelope function decomposition that we detail in the following section.
\subsection{Envelope functions}
\par
In the following, we shall use the symbol $\cF$ to denote the Fourier transformation 
on $L^2(\mR^d)$ 
\begin{equation}
\label{fourier}
\cF \psi (k) = {1\over (2\pi)^{d/2}} \int_{\mR^d}  e^{-i x\cdot k } \psi(x) \, dk
\end{equation}
and $\cF^* = \cF^{-1}$ for the inverse transformation.
We shall use a hat, $\hat \psi = \cF\psi$, for the Fourier transform of $\psi$.
\begin {definition}
We define $L^2_\cB(\mR^d) \subset L^2(\mR^d)$ to be  the subspace
of $L^2$-functions supported in $\cB$:
\begin{equation}
\label{L2Bdef}
  L^2_\cB(\mR^d) = \left\{ f \in L^2(\mR^d)\  \left| \ 
  \supp \big(f\big) \subset \cB \right.\right\}.
\end{equation}
Thus, $\cF^*L^2_\cB(\mR^d)$ is the space of $L^2$-functions whose Fourier transform 
is supported in $\cB$.
\end {definition}
The envelope function decomposition is defined by the following theorem.
\begin{theorem}
\label{T1}
Let $v_n : \mR^d \to \mC$ be $\cL$-periodic functions such that
$\{ v_n \mid n \in \mN\}$ is an orthonormal basis of $L^2(\cC)$.
For every $\psi \in L^2(\mR^d)$ there exists a unique sequence
$\{ f_n \in \cF^*L^2_\cB(\mR^d)\mid n \in \mN\}$ such that
\begin{equation}
\label{EFdec}
 \psi = \abs{\cC}^{1/2}\,\sum_n f_n\,v_n.
\end{equation}
We shall denote $f_n = \pi_n(\psi)$. The decomposition satisfies the Parseval identity
\begin{equation}
\label{Parseval}
 \bk{\psi,\varphi}_{L^2(\mR^d)} = \sum_n \bk{\pi_n(\psi), \pi_n(\phi)}_{L^2(\mR^d)}.
\end{equation}
For any $\eps > 0$ we shall consider the scaled version $f_n^\eps = \pi_n^\eps (\psi)$ of the envelope 
function decomposition as follows:
\begin{equation}
\label{EF}
  \psi(x) = \abs{\cC}^{1/2}\,\sum_n f_n^\eps(x) \,v_n^\eps(x),
\end{equation}
with $\hat f_n^\eps  \in L^2_{\cB/\eps}(\mR^d)$, where  
\begin{equation}
\label{vScaled}
  v_n^\eps(x) = v_n\left( \frac{x}{\eps} \right).
\end{equation}
We still have the Parseval identity
\begin{equation}
\label{Parsevaleps}
 \bk{\psi,\varphi}_{L^2(\mR^d)} = \sum_n \bk{\pi_n^\eps (\psi), \pi_n^\eps (\varphi)}_{L^2(\mR^d)}.
\end{equation}
Finally, the Fourier transforms  of the $\eps$-scaled  envelope functions are given by 
\begin{equation}
\label{EFChiEps}
 \hat f_n^\eps(k) = \int_{\mR^d} \ol \cX_{n,k}^\eps (x)\,\psi(x)\,dx,
\end{equation}
where, for $x \in \mR^d$, $k \in \mR^d$, $n \in \mN$, 
\begin{equation}
\label{ChiEps}
 \cX_{n,k}^\eps(x) = \abs{\cB\,}^{-1/2}\,\cara_{\cB/\eps}(k) \,\e^{ik\cdot x}\,v_n^\eps(x).
\end{equation}
\end{theorem}
The proof of this theorem is postponed  to Section \ref{post}.
\begin{remark}
 Note that the above result is a variant of the so-called Bloch transform. In \cite{allaire08}, the function
$$\widehat{\psi} (x,k) = \abs{\cC}^{1/2}\,\sum_n \widehat f_n (k) \,v_n(x)$$
is referred to as the Bloch transform of $\psi$. We  also refer to \cite{allaire98},  \cite{ReedSimonIV78} 
and \cite{kuch93} for Bloch wave methods in periodic media.
\end{remark}
\begin {definition}
The functions $f_n = \pi_n(\psi)$ of Theorem \ref{T1} will be called the {\em envelope functions} 
of $\psi$ relative to the basis $\{ v_n \mid n\in\mN\}$, while $f_n^\eps = \pi_n^\eps(\psi)$ will be called the 
$\eps$-scaled envelope function relative  to the basis $\{ v_n \mid n\in\mN\}$
\end {definition}
\begin{theorem}
\label{T2}
Let us consider the $\eps$-scaled envelope function decomposition \eqref{EF} of $\psi \in L^2(\mR^d)$.
Then, for every $\theta \in L^1(\mR^d)$ such that $\hat\theta \in L^1(\mR^d)$, we have
\begin{equation}
 \lim_{\eps\to 0} \int_{\mR^d} \theta(x)\left[|\psi(x)|^2 -\sum_{n} |f_n^\eps(x)|^2\right] dx = 0.
\end{equation}
\end{theorem}
The proofs of this theorem is also postponed to Section \ref{post}.
\subsection{Functional spaces}
In this section, we define some functional spaces which will be used all along the paper.
\begin {definition}
\label{Spaces}
We define the space  $\cL^2  = \ell^2\left(\mN,L^2(\mR^d)\right)$ as the Hilbert space of sequences 
$g = (g_0,g_1,\ldots)$, $g_n = g_n(k)$,  with $g_n \in L^2(\mR^d)$, such that 
\begin{equation}
  \norma{g}_{\cL^2}^2 = \sum_n \norma{g_n}^2_{L^2(\mR^d)} < \infty.
\label{L2}
\end{equation}
Moreover, for $\mu \geq 0$ let $\cL^2_\mu$ be the subspace of all sequences $g \in \cL^2$ 
such that
\begin{equation}
\label{L2mu}
 \norma{g}_{\cL^2_\mu}^2 = \norma{(1+\abs{k}^2)^{\mu/2}g}_{\cL^2}^2 
 =  \sum_{n} \norma{(1+\abs{k}^2)^{\mu/2}g_n}_{L^2}^2< \infty
\end{equation} 
and let $\cH^\mu =  \ell^2\left(\mN,H^\mu(\mR^d)\right)$, with
\begin{equation}
\label{Hmu}
 \norma{f}_{\cH^\mu}^2 = \sum_{n} \norma{f_n}_{\cH^\mu}^2 
 =  \sum_{n} \norma{(1+\abs{k}^2)^{\mu/2}\widehat{f}_n}_{L^2}^2< \infty.
\end{equation} 
It is readily seen that $f\in \cH^\mu$ if and only if $\widehat{f}\in \cL^2_\mu$.
Let us redefine the eigenpairs $(E_n,v_n)$ of the operator $H_\cL^1 = -\frac{1}{2}\Delta + W_\cL$ 
with periodic boundary conditions by
\begin{equation}
\label{periodic}
\left\{
\begin{aligned}
&-\textstyle{\frac{1}{2}}\Delta v_n + W_\cL v_n = E_n v_n, \quad \text{\rm on $\cC$}
\\[6pt]
&\int_\cC |v_n|^2\, dx = 1, \quad \text{\rm $v_n$ periodic}
\end{aligned}
\right.
\end{equation} 
(note that $v_n = u_{n,0}$, according to Definition \ref{BlochDef}).
The sequence $E_n$ is increasing and tends to $+\infty$.
\end {definition}
Let us now define the functional spaces for the external potential:
\begin{equation}
\label{vdeff}
  \cW_\mu = \Big\{ V \in L^\infty(\mR^{2d}) \ \Big| \ 
  V(\cdot,z+\lambda) = V(\cdot,z),\ \lambda \in \cL,\
  \norma{V}_{\cW_\mu} < \infty \Big\},
\end{equation}
where
\begin{equation}
\label{Mdef}
  \norma{V}_{\cW_\mu} = \frac{1}{(2\pi)^{d/2}}\, \mathop\mathrm{ess\,sup}\limits_{z \in \cC}
 \int_{\mR^d} (1 + \abs{k})^\mu |\hat V(k,z)|\,dk
\end{equation}
and $\ds \widehat{V}(k,z) = (2\pi)^{-d/2}\int_{\mR^d} \e^{-ik\cdot x}\, V(x,z)\,dx$.
\par
We finally define for any positive constant $\gamma$  the truncation operator 
\begin{equation}
\label{truncation}
\cT_\gamma (f) = \cF^* (\cara_{\gamma \cB} \widehat{f}).
\end{equation}
It is now readily seen that the truncation operator satisfies for any nonnegative real numbers $s,\mu$,
\begin{equation}
\norma{f-\cT_\gamma f}_{H^s} \leq C \gamma^{-\mu} \norma{f}_{H^{s+\mu}},
\label{truc-error}
\end{equation}
where $C>0$ is a suitable constant independent of $\gamma$. 
\subsection{Main Theorem}
We announce in this section the main theorem of our paper. 
We recall that $(v_n,E_n)$ are defined by \eqref{periodic}.
\begin{theorem}
\label{main}
Assume that $W_\cL \in L^\infty$ and that all the eigenvalues $E_n = E_n(0)$ are simple.
Let $\psi^{in,\eps}$ be an initial datum in $L^2(\mR^d)$, let $f^{in,\eps}_n= \pi_n^\eps(\psi^{in,\eps})$ 
be its scaled envelope functions relative to the basis $v_n$. 
Assume that the sequence $f^{in,\eps}$ belongs to $\cH^\mu$, with a uniform bound for the norm 
as $\eps$ vanishes,  and  that it converges in $\cL^2$ as $\eps$ tends to zero to an initial datum $f^{in}$. 
Let $\psi^\eps$ be the unique solution of 
\begin{equation}
\label{SE1}
\begin{aligned}
  &i\partial_t\, \psi^\eps(t,x) = \left( -\frac{1}{2}\Delta + \frac{1}{\eps^2}
   \,W_\cL\left(\frac{x}{\eps}\right) + V\left(x,\frac{x}{\eps}\right) \right) \psi^\eps(t,x),
   \\
   &\psi(t=0) = \psi^{in,\eps},
   \end{aligned}
\end{equation}
and assume that $V\in \cW_\mu$ for a positive $\mu$.
Then for any $\theta\in L^1(\mR^d)$ such that $\widehat{\theta} \in L^1(\mR^d)$, we have the following 
local uniform convergence in time 
$$
  \int |\psi^\eps(t,x)|^2\theta(x)\, dx \to \int \sum_n |h_n(t,x)|^2 \theta(x)\, dx
$$  
where the envelope function $h_n$ is the unique solution of the homogenized Schr\"od\-in\-ger equation
$$  
i\partial_t\, h_n =   
  - \frac{1}{2} \DIV \left( \mM_n^{-1} \nabla h_n \right)
  + V_{nn}(x)\,h_n,\quad h_n(t=0) = f^{in}_n,
$$
with
$$
  V_{nn} = \int_\cC V(x,z)|v_n(z)|^2\, dz
$$
and
$$
  \mM_n^{-1}  = \nabla\otimes\nabla \,E_n(k)_{\,|k=0} 
  = I - 2 \sum_{n' \not= n} \frac{P_{nn'} \otimes P_{n'n}}{E_n - E_{n'}}.
$$
(effective mass tensor of the $n$-th band).
\end{theorem}
\section{From the Schr\"odinger equation to the  k$\cdot$p model}
\label{sec3}
Let $\psi^\eps(t,x)$ be the solution of the Schr\"odinger equation \eqref{SE1} and let $f_n^\eps(t,x)$ be its 
$\eps$-scaled envelope function relative to the basis $v_n$  defined in 
\eqref{periodic} and \eqref{vScaled}: 
$$
  \psi^\eps(t,x) = \abs{\cC}^{1/2}\sum_n  f_n^\eps(t,x) v_n^\eps(x)
$$
Let us define $$g_n^\eps(t,k) = \widehat{f}_n(t,k).$$
\par
From now on, we will reserve the notation $f$ for functions of the position variable $x$, while $g$ will be 
used for functions of the wavevector $k$.
Multiplying the Schr\"odinger equation by $\overline{\cX_{n,k}^\eps (x)}$ (see Eq.~\eqref{ChiEps}) 
and integrating over $k$ leads  to the following equation 
\begin{equation}
\label{EXE2}
\begin{aligned}
  i\partial_t \,g_n^\eps(t,k) 
  &= \frac{1}{2}\,\abs{k}^2\,g_n^\eps(t,k) - \frac{i}{\eps} \sum_{n'} 
     k\cdot P_{nn'} g_{n'}^\eps(t,k) + \frac{1}{\eps^2}\,E_n\,g_n^\eps(t,k)
\\
  &+ \sum_{n'} \int_{\mR^d} U^\eps_{nn'}(k,k')\,g_{n'}^\eps(t,k') \,dk',
\end{aligned}
\end{equation}
where the kernel $U_{nn'}(k,k')$ is given by
$$
\begin{array}{lll}
 \ds U^\eps_{nn'}(k,k') &= &\ds \int_{\mR^d} \ol\cX_{n,k}^\eps(x)\, V\left(x,{x\over\eps}\right)\,\cX_{n',k'}^\eps(x)\,dx
\\[8pt]
 &=&  \ds  \abs{\cB\,}^{-1}\,\cara_{\cB/\eps}(k)   \int_{\mR^d}  \cara_{\cB/\eps}(k') \,\e^{-i(k-k')\cdot x}\,
 \overline{v_n}^\eps(x)\, V\left(x,{x\over\eps}\right)\,v_{n'}^\eps(x)\,dx.
  \end{array}
$$
By writing 
$$
  V(x,z) v_n(z) = \sum_{n'} V_{n'n}(x) v_{n'} (z),
$$
where
\begin{equation}
\label{VnmDef}
 V_{n'n}(x) = \int_\cC \ol v_{n'}(z)\, v_n(z)\, V(x,z)\,dz = \ol V_{nn'}(x),
\end{equation}
we can express $U^\eps_{nn'}(k,k')$ in the form
 \begin{equation}
\label{uneps}
 U^\eps_{nn'}(k,k') = \frac{\cara_{\cB/\eps}(k)}{\abs{\cB}}  \sum_{m} \int_{\mR^d}  \cara_{\cB/\eps}(k') 
 \,\e^{-i(k-k')\cdot x}\,\overline{v_n}^\eps(x)\, V_{mn'}(x) v_{m}^\eps(x)\,dx
\end{equation}
In position variables, the envelope functions satisfy the system
\begin{multline}
\label{EXE}
  i\partial_t \,f_n^\eps(t,x) = 
 {E_n\over \eps^2} f_{n}(t,x)
  - \textstyle{\frac{1}{2}} \Delta\,f_n^\eps(t,x) 
\\
  - \frac{1}{\eps} \sum_{n'\in\mN} P_{nn'}\cdot\nabla f_{n'}^\eps(t,x) 
  + \sum_{n'\in\mN} \int_{\mR^d} V_{n n'}^\eps(x,x')\, f_{n'}^\eps(t,x')\,dx',
\end{multline}
where
\begin{multline}
  V_{nn'}^\eps (x,x') = 
  \frac{1}{(2\pi)^d\abs{\cB}}\,  \int_{\cB/\eps} dk \int_{\mR^d}dy  \int_{\cB/\eps} dk'  \times
\\
  \times \left\{ \e^{ik\cdot x}   \e^{-i(k - k')\cdot y}\, \ol v_n^\eps(y)  
  V\left(y,{y\over\eps}\right) v_{n'}^\eps(y)\,\e^{-ik'\cdot x'}  \right\}
\end{multline}
From equation \eqref{EXE} we see that the fast oscillation scales are different for different envelope functions. 
This will naturally lead to adiabatic decoupling (see  \cite{hagedorn-joye,panati,spohn-teufel,teufel}).
\begin {definition}
\label{U}
Let us define the operator $\cU^\eps$ on $\cL^2$ as follows: for any element 
$g =(g_0, g_1, \ldots)$ of $\cL^2$
\begin{equation}
\label{Udef}
  \left( \cU^\eps g \right)_n(k) = 
  \sum_{n'} \int_{\mR^d} U^\eps_{nn'}(k,k')\,g_{n'}^\eps(k') \,dk'.
\end{equation}
Let us also define the operator $\bV^\eps$ on the position space $\cL^2$ by
\begin{equation}
\label{Vdef}
  \left( \bV^\eps f \right)_n(x) = 
  \sum_{n'} \int_{\mR^d} V^\eps_{nn'}(x,x')\,f_{n'}^\eps(x') \,dx'.
\end{equation}
We obviously have 
$$
  \widehat{\bV^\eps(f)} = \cU^\eps(\widehat{f}).
$$
\end {definition}
Since $v_n$ and $v_m$ are $\cL$-periodic, the formal limit of $ U^\eps_{nn'}(k,k')$ is  given by 
$$
  U^0_{nn'}(k,k') 
  = \sum_{m} { \bk{v_n,v_m}\over \abs{\cB} \abs{\cC}}\int_{\mR^d} \e^{-i(k-k')\cdot x}\, V_{mn'}(x)\, dx 
  =   {1 \over (2\pi)^{d/2} } \widehat{V_{nn'}}(k-k').
 $$
Therefore the formal limit of  $\cU^\eps$ is the operator $\cU^0$ defined by
 \begin{equation}
\label{U0def}
  \left( \cU^0 g \right)_n(k) = 
  \sum_{n'} \frac{1}{(2\pi)^{d/2}}\int_{\mR^d}\hat V_{nn'}(k-k')\,g_{n'}(k')\,dk',
\end{equation}
which means that the in position space the limit of $\bV^\eps$ is the non diagonal multiplication operator 
$\bV^0$ defined by 
\begin{equation}
\label{V0def}
  \left( \bV^0 f \right)_n(x) = 
  \sum_{n'} V_{nn'}(x)\,f_{n'}(x).
\end{equation}
The operators become diagonal in $n$ if $V(x,z)$ does not depend on $z$. 
Indeed, in this case $V_{nn'}(x) = V(x)\delta_{nn'}$.
The k$\cdot$p approximation found in semiconductor theory 
\cite{Wenckebach99}, consists in replacing  the operator $\cU^\eps$ by $\cU^0$. 
Let us now analyze the departure of $\cU^\eps$ from $\cU^0$.
\begin{lemma}
\label{Lemma0}
Let the external potential $V(x,z)$ be in $L^\infty$.
Then, for any $\eps \geq 0$, $\cU^\eps$ is a bounded operator on $\cL^2$ and we have 
the uniform bound 
\begin{equation}
\label{Knorm}
 \norma{\cU^\eps} \leq \norma{V}_{L^\infty},  \quad \forall\ \eps \geq  0.
\end{equation}
\end{lemma}
\begin{proof}
Let us begin with the case $\eps = 0$. 
We remark that 
$$
  \cU^0 g = \widehat{\bV ^0 (f)},
$$
where $f = \cF^*(g)$. 
Let $G$  be another element of $\cL^2$, and let  $F$ be its back Fourier transform. 
We have
$$
\begin{aligned}
 \abs{\bk{\cU^0 g,G}} 
  & =  \abs{\bk{\bV^0 f, F}} = \abs{\sum_{nn'} \int V_{nn'}(x) f_{n'}(x) \overline{F_n}(x)\, dx} 
  \\[6pt]
  &= \abs{\sum_{nn'} \int V(x,z) v_{n'}(z)\overline{v_n}(z) f_{n'}(x)\overline{F_n}(x)\, dx\, dz} 
  \\[6pt]
  &=  \abs{\int  V(x,z) \left[\sum_n f_{n}(x) v_n(z) \right]\overline{\left[\sum_n F_{n}(x) v_n(z)\right]} dx\, dz}
  \\[2pt]
  &\leq \norma{V}_{L^\infty}\!\!  \left[\int  \Big| \sum_n f_{n}(x) v_n(z)\Big|^2 dx\, dz\right]^{\! {1\over 2}}
    \! \left[\int  \Big|\sum_n F_{n}(x) v_n(z)\Big|^2 dx\, dz \right]^{\! {1\over 2}}
   \\[6pt]
&\leq  \norma{V}_{L^\infty} \norma{f}_{\cL^2}\norma{F}_{\cL^2} = \norma{V}_{L^\infty} \norma{g}_{\cL^2}\norma{G}_{\cL^2}. 
\end{aligned}
$$
Since the result holds for any $g$ and $G$ in $\cL^2$, this implies that 
$\norma{\cU^0(g)}_{\cL^2} \leq \norma{V}_{L^\infty} \norma{g}_{\cL^2}.$
For $\eps > 0$ it is enough to observe that $\cU^\eps$ is unitarily equivalent to the multiplication operator 
by $V(x,{x\over\eps})$ in position space. 
More precisely, defining $f^\eps(x) = \cF^*(\carB g)$ and defining 
$\psi^\eps(x) = \sum_n f_n^\eps (x) v_n^\eps(x)$ so that $f_n^\eps = \pi_n^\eps(\psi^\eps)$, 
then it follows from the definition of $\cU^\eps$ that 
$$
  ( \cU^\eps g)_n  =\cF\left[ \pi_n^\eps \left( V\Big(x,{x\over\eps}\Big) \psi^\eps\right)\right].
$$
It is now readily seen that
$$
  \norma{\cU^\eps(g)}_{\cL^2}^2 = \norma{V\left(x,{x\over\eps}\right) \psi^\eps}_{L^2}^2 
  \leq \norma{V}_{L^\infty}^2 \norma{ \psi^\eps}_{L^2}^2
 \leq \norma{V}_{L^\infty}^2 \norma{g}_{\cL^2}^2.
$$
\end{proof}
\begin{lemma}
\label{gammab}
For any $\gamma > 0$ let $\gamma \cB$ be the set of $\gamma k$ where $k$ is in $\cB$. Then 
$$\gamma \cB + \beta \cB = (\gamma+\beta )\cB.$$
Moreover Let 
$k\in \cB$ and  $k'\in {1\over 3} \cB$. Let $\lambda$ a non vanishing element of the reciprocal lattice $\cL^*$. 
Then $k- k' + \lambda \notin {1\over 3} \cB$.
\end{lemma}
The proof of this lemma is immediate (using the fact that $\cB$ is the linear deformation of a hypercube,
see definition \eqref{Brillo}) and is left to the reader.
\begin{lemma}
\label{Lemma01}
Let $V \in \cW_0$ and $g \in \cL^2$ be such that 
$\supp\big(\hat V_{nm}\big) \subset {1\over 3 \eps}  \cB$ 
and $\supp(g_n) \subset {1\over 3 \eps}  \cB$, for all $n, m \in \mN$.
Then, in this case, $\,\cU^\eps g = \cU^0 g$.
\end{lemma}
\begin{proof}
Let us first notice that  $\{ \abs{\cC}^{-1/2}\,\e^{i\eta\cdot x} \mid \eta \in \cL^* \}$
is a orthonormal basis of $L^2(\cC)$ (the Fourier basis).
We first deduce from \eqref{uneps} and from the identity
$$
 v_n (y) = {1\over \abs{\cC}^{1/2}} \sum_{\lambda \in \cL^*} v_{n,\lambda} e^{i\lambda\cdot x}
$$
where $v_{n,\lambda} = \langle v_n, { e^{i\lambda\cdot x}\over \abs{\cC}^{1/2}}\rangle$  that 
\begin{multline*}
\begin{aligned}
  \left( \cU^\eps  g \right)_n(k)   =
 \sum_{\lambda,\lambda'\in \cL^*}\sum_{m,n'} \int_{\mR^d\times \mR^d}
   &e^{-i(k-k' +{\lambda-\lambda' \over \eps})\cdot x} \carB (k') \carB(k) \times
   \\
   &\times V_{nm}(x) \overline{v_{m,\lambda} } \,v_{n',\lambda'}\,   g_{n'} (k')\, dx\, dk' =
   \end{aligned}
\\[6pt]
  \carB(k) {(2\pi)^{d/2}} \!\!\! \sum_{\lambda,\lambda'\in \cL^*}
 \sum_{m,n'} \int_{\cB/\eps} \widehat{V}_{nm} \left(  k-k' +\textstyle{{\lambda-\lambda' \over \eps}}\right) 
  \overline{v_{m,\lambda} } \,v_{n',\lambda'} \,   g_{n'} (k')\, dx\, dk'.
\end{multline*}
Since the support of $g_{n'}$ is included in $\cB/3\eps$ and $k\in \cB/\eps$, Lemma \ref{gammab} 
implies that the only contributing terms to the above sum are those for which $\lambda = \lambda'$. 
Therefore, we are lead to evaluate  $\sum_{\lambda} \overline{v_{m,\lambda} } \,v_{n',\lambda}$ 
which is equal to $\bk{v_{n'}, v_m} = \delta_{mn'}$ because of the orthonormality of the family $(v_n)$. 
Therefore
 $$
  \left( \cU^\eps  g \right)_n(k)  = (2\pi)^{-d/2} \carB(k)\sum_{n'} \int_{\cB/\eps} 
  \widehat{V_{nn'}} (  k-k' )  g_{n'} (k')\, dx\, dk'.
$$
Now, we can remove   $\carB(k)$ from the right hand side of the above identity, since both the support of 
$g_{n'}$ and  that of $\widehat{V_{nn'}}$ are in ${1\over 3\eps} \cB$. 
Hence
$$
  \left( \cU^\eps  g \right)_n(k)  = (2\pi)^{-d/2}  \sum_{n'} \int_{\mR^d}  
  \widehat{V_{nn'}} (  k-k' )  g_{n'} (k')\, dx dk'  = \left( \cU^0  g \right)_n(k).
$$
\end{proof}
\begin{theorem}
\label{T3}
Assume that $V \in \cW_\mu$ for some $\mu \geq 0$. 
Then, a constant $c_\mu > 0$, independent of $\eps$, exists such that
\begin{equation}
\label{Kextima}
  \norma{\cU^\eps  g - \cU^0 g }_{\cL^2} \leq \eps^\mu\,c_\mu\,\norma{V}_{\cW_\mu} \,\norma{g}_{\cL^2_\mu}
\end{equation}
for all $g \in \cL^2_\mu$ and for all $\eps > 0$.
\end{theorem}
\begin{proof}
Let the smoothed potential $V_s^\eps$ be defined by
\begin{equation}
\label{UepsDef}
 \hat V_s^\eps(k,z) = \cara_{\cB/3\eps}(k)\,\hat V(k,z).
\end{equation}
Moreover, let $\cU^\eps_s$ denote the operator $\cU^\eps$ with the potential $V_s$.
Let us assume firstly that $\supp\,(g_n) \subset \cB/3\eps$ for all $n \in \mN$.
Then, from Lemma \ref{Lemma01} we have $\cU^\eps_s g = \cU^0_s g$ 
and we can write 
\begin{equation}
\label{aux0}
  \norma{\cU^\eps g - \cU^0 g}_{\cL^2} \leq
  \norma{\cU^\eps g - \cU^\eps_s g}_{\cL^2}  + \norma{\cU^0_s g - \cU^0 g}_{\cL^2} .
\end{equation}
Using \eqref{Knorm} and the linearity of $\cU^\eps$ and $\cU^0$ with respect to the potential, 
we have
$$
 \norma{\cU^\eps g - \cU^\eps_s g}_{\cL^2} \leq  \norma{V-V_s^\eps}_{\cW_0} \, \norma{g}_{\cL^2}.
 \qquad \eps \geq 0,
$$ 
Recalling the definition \eqref{Mdef}, we also have
$$ 
 \norma{V-V_s^\eps}_{\cW_0} =   \frac{1}{(2\pi)^{d/2}}\, \mathop\mathrm{ess\,sup}\limits_{z \in \cC}
 \int_{\mR^d \setminus \cB/3\eps}  |\hat V(k,z)|\,dk
$$ $$
 \leq \frac{1}{(2\pi)^{d/2}}\, \mathop\mathrm{ess\,sup}\limits_{z \in \cC}
 \int_{k\notin \cB/3\eps} \left( \frac{\abs{3\eps k} }{R} \right)^\mu\, |\hat V(k,z)|\,dk
 \leq \left(\frac{3\eps}{R}\right)^\mu \, \norma{V}_{\cW_\mu}
$$
where $R>0$ is the radius of a sphere contained in $\cB$.
Then (still in the case $\supp(g_n) \subset \cB/3\eps$),
from \eqref{aux0} we get
\begin{equation}
\label{I2}
  \norma{\cU^\eps g - \cU^0 g}_{\cL^2}  
  \leq 2\left(\frac{3\eps }{R}\right)^\mu  \norma{V}_{\cW_\mu}\,\norma{g}_{\cL^2}.
\end{equation}
Now, if $g \in {\cL^2_\mu}$ (Definition \ref{Spaces}), we can write (using $\cara^c = 1- \cara$)
\begin{equation}
\label{aux2}
  \norma{\cU^\eps g - \cU^0 g}_{\cL^2}  \leq  
  \norma{\cU^\eps \cara_{\cB/3\eps}^c g}_{\cL^2} +
  \norma{(\cU^\eps - \cU^0) \cara_{\cB/3\eps} g}_{\cL^2} 
  + \norma{\cU^0 \cara_{\cB/3\eps}^c g}_{\cL^2} 
\end{equation}
From \eqref{Knorm} we have 
$\norma{\cU^\eps \cara_{\cB/3\eps}^c g}_{\cL^2} \leq \norma{V}_{\cW_0} \norma{\cara_{\cB/3\eps}^c g}_{\cL^2}$,
for all $\eps \geq 0$.
But
$$
  \norma{\cara_{\cB/3\eps}^c g}_{\cL^2}^2   = \sum_n \int_{k\notin  \cB/3\eps} \abs{g_n(k)}^2 \,dk
$$ $$
 \leq \sum_n \int_{k\notin \cB/3\eps} 
 \left( \frac{\abs{3\eps k} }{R} \right)^{2\mu} \abs{g_n(k)} \,dk
 \leq \left(\frac{3\eps}{R}\right)^{2\mu}\norma{g}_{{\cL^2_\mu}}^2
$$
and so we can estimate the first and third term in the right hand side of \eqref{aux2} 
as follows:
$$
  \norma{\cU^\eps \cara_{\cB/3\eps}^c g}_{\cL^2} + \norma{\cU^0 \cara_{\cB/3\eps}^c g}_{\cL^2} 
  \leq 2 \left(\frac{3\eps}{R}\right)^{\mu} \norma{V}_{\cW_0}\,\norma{g}_{{\cL^2_\mu}}.
$$ 
Moreover, since Eq.~\eqref{I2} holds for $\cara_{\cB/3\eps} g$, then we can
estimate also the second term:
$$
  \norma{(\cU^\eps - \cU^0) \cara_{\cB/3\eps} g}_{\cL^2} \leq 
  2 \left(\frac{3\eps }{R}\right)^\mu \norma{V}_{\cW_\mu}\,\norma{g}_{\cL^2}.
$$
Since $\norma{V}_{\cW_0} \leq \norma{V}_{\cW_\mu}$ and $\norma{g}_{\cL^2} \leq \norma{g}_{{\cL^2_\mu}}$,
then from \eqref{aux2} we conclude that \eqref{Kextima} holds, with 
$c_\mu = 4 (3/R)^\mu$ (note that $R$ does not depend on $\eps$).
\end{proof}
\section{Diagonalization of the k$\cdot$p Hamiltonian}
\label{sec4}
In this section, we consider the case $V(x,z) = 0$ and concentrate on the diagonalization 
of the k$\cdot$p Hamiltonian. 
The envelope function dynamics are then given in Fourier variables by Eq.~\eqref{EXE2}
which we rewrite under the form
\begin{equation}
\label{FKPE}
  i\eps^2\partial_t \,g_n(t,k) = 
  \frac{1}{2}\eps^2 \abs{k}^2 g_n(t,k)
  - i \eps \sum_{n'} k \cdot P_{nn'} g_{n'}(t,k) +  E_n g_n(t,k).
\end{equation}
Putting $\xi = \eps k$, we are therefore led to consider, for any fixed $\xi \in \mR^d$, 
the following operators, acting in $\ell^2 \equiv \ell^2(\mN,\mC)$ and defined on their 
maximal domains:
\begin{equation}
\label{A012def}
  (A_0)_{nn'} = E_n \delta_{nn'},
\quad
  \left(A_1(\xi)\right)_{nn'} = -i\xi\cdot P_{nn'},
\quad
  \left(A_2(\xi)\right)_{nn'} = \frac{1}{2}\,\abs{\xi}^2\,\delta_{nn'}.
\end{equation}
Moreover, we put $A(\xi) = A_0 + A_1(\xi) + A_2(\xi)$, so that
\begin{equation}
\label{Adef}
   \left(A(\xi)\right)_{nn'}  =  E_n \delta_{nn'} -i\xi\cdot P_{nn'} 
   + \frac{1}{2}\,\abs{\xi}^2\,\delta_{nn'}
\end{equation}
is the operator at the right-hand side of Eq.~\eqref{FKPE} (with $\xi = \eps k$).
\begin{lemma}
\label{P2}
The following properties hold:
\begin{enumerate}
\item[\rm (a)]
  for any given $\xi \in \mR^d$, $A_1(\xi)$ is $A_0$-bounded with $A_0$-bound less than 1, 
  which implies that $A(\xi) =  A_0 + A_1(\xi) + A_2(\xi)$ is self-adjoint on the (fixed) domain 
  of $A_0$,  that is
\begin{equation}
\label{D0}
   \cD(A_0) = \Big\{ g \in \ell^2\ \Big|\ \sum_n \abs{E_n g_n}^2 < \infty \Big\};
\end{equation}
\item[\rm (b)]
  $\{A(\xi) \mid \xi \in \mR^d \}$ is a holomorphic family of type (A) of self-adjoint 
  operators \cite{Kato80};
\item[\rm (c)]
  for any given $\xi \in \mR^d$, $A(\xi)$ has compact resolvent, which implies that $A(\xi)$
  has a sequence of eigenvalues $\lambda_1(\xi) \leq \lambda_2(\xi) \leq \lambda_3(\xi) \leq \cdots$, 
  with $\lambda_n(\xi) \to \infty$, and a corresponding sequence 
  $\varphi^{(1)}(\xi)$, $\varphi^{(2)}(\xi)$, $\varphi^{(3)}(\xi) \ldots$ of orthonormal eigenvectors .
\end{enumerate}
\end{lemma}
\begin{proof}
(a)
We first recall (see \eqref{periodic}) that $(v_n,E_n)$ is an eigencouple of 
$H_\cL^1= -\frac{1}{2}\Delta + W_\cL$ on the domain $\rH_\per^2(\cC)$
(the subscript ``per'' denoting periodic boundary conditions).
The operator $A_0$ is the representation in the basis $(v_n)$ of the operator $H^1_\cL$,
while  $A_1(\xi)$ is the representation in the same basis of $-i\xi\cdot\nabla$ with
domain $\rH^1(\cC)$:
$$
 \cD\left(A_0\right) \equiv \rH_\per^2(\cC) \subset \rH^1(\cC) 
 \equiv \cD\left(A_1(\xi)\right).
$$
Then, for any given sequence $(g_n)$, denoting $g(x) = \sum_n g_n v_n(x)$, we have 
$$
 \frac{1}{2}\int_\cC \abs{\nabla g(x)}^2\,dx + \int_\cC W_\cL(x) \abs{g(x)}^2\,dx 
 = \bk{ H^1_\cL g, g}_{L^2(\cC)}  = \sum_n E_n\abs{g_n}^2.
$$
Since $W_\cL$ is bounded and  $W_\cL \geq 1$, then 
for $g \in \cD(A_0)$ we obtain 
\begin{equation}
\label{A1extim}
  \norma{A_1(\xi) g}_{\ell^2}^2 \leq \abs{\xi}^2\norma{\nabla g}_{L^2(\cC)}^2
  \leq 2\abs{\xi}^2 \sum_n E_n \abs{g_n}^2,
\end{equation}
where we used the notation $g$ for both $g(x) = \sum_n g_n v_n(x)$ and for the sequence $g = (g_n) \in \ell^2$.
Since $E_n \to \infty$, then, for any given $0 < b < 1$, 
a positive integer $n(\xi)$ exists 
such that $2\abs{\xi}^2 E_n < b E_n^2$ for $n \geq n(\xi)$ and we can write
$$
  2\abs{\xi}^2 \sum_n E_n \abs{g_n}^2 \leq
  2\abs{\xi}^2 E_{n(\xi)} \sum_{n=1}^{n(\xi)}\abs{g_n}^2 
  + \sum_{n=n(\xi)}^\infty b\abs{E_ng_n}^2.
$$ 
Thus, $\norma{A_1(\xi) g}_{\ell^2}^2 \leq 2\abs{\xi}^2 E_{n(\xi)} \norma{g}_{\ell^2}^2 + b\, \norma{A_0g}_{\ell^2}^2$,
with $b<1$, which proves point (a).
The proof of the remaining points is standard (see Refs.~\cite{BerezinShubin91,Kato80,ReedSimonIV78}).
\end{proof}
\begin{remark}
Recalling Definition \ref{BlochDef} and Eq.~\eqref{HLK} we see that $A(\xi)$ is nothing but the expression of the
fiber Hamiltonian $H_\cL(\xi)$ in the Bloch basis $v_n = u_{n,0}$.
Then, the diagonalization of $A(\xi)$ corresponds to the diagonalization of $H_\cL(\xi)$
and, therefore, the eigenvalues $\lambda_n(\xi)$ coincide with the energy bands $E_n(\xi)$ inside the
Brillouin zone.
Moreover, $\varphi^{(n)}(\xi)$ is clearly the component expression of $u_{n,\xi}$ 
in the basis $u_{n,0}$, i.e.\ $\varphi^{(n)}(\xi) = \bk{u_{n,\xi}, u_{n,0}}_{L^2(\cC)}$.
\end{remark}
The eigenvalues $\lambda_n(\xi)$ have been numbered in increasing order for each $\xi$; this means that, 
when a eigenvalue crossing occurs, then the smoothness of $\lambda_n(\xi)$ (and of $\varphi^{(n)}(\xi)$) 
is lost.
However, since we are assuming that $\lambda_n(0) = E_n$ are simple, then $\lambda_n(\xi)$ and 
$\varphi^{(n)}(\xi)$ are analytic in a neighborhood of the origin. 
Of course, such neighborhood depends of $n$.
Next lemma allows to estimate the growth of the eigenvalues and, consequently, the size of the
analyticity domain.
\begin{lemma}
\label{P4}
For any given $\xi \in \mR^d$, an integer $n_0(\xi) \geq 0$ exists such that
\begin{equation}
\label{LambdaGrowth}
  \abs{\lambda_n(\xi) - E_n} \leq \abs{\xi} \sqrt{2 E_n} + \frac{1}{2}\abs{\xi}^2,
\quad
 \text{for all $ n \geq n_0(\xi)$.}
\end{equation}
\end{lemma}
\begin{proof}
The behavior of the eigenvalues $\lambda_n(\xi)$ for large $n$ will be investigated 
by means of the {\em max\,min} principle, which holds for increasingly-ordered eigenvalues, 
\cite{ReedSimonI72}.
Since the operators $A(\xi)$ have compact resolvent, the {\em max\,min} principle 
reads as follows:
\begin{equation*}
  \lambda_n(\xi) = \max_{S \in M_{n-1}} \; \min_{g \in S^\perp\cap\cD(A_0),\ \norma{g}=1 } 
  \,\bk{A(\xi)g, g}_{\ell^2},
\end{equation*}
where $M_n$ denotes the set of all subspaces of dimension $n$. 
In particular,
\begin{equation*}
  \lambda_n(0) = E_n  = \max_{S \in M_{n-1}} \; \min_{g \in S^\perp\cap\cD(A_0),\ \norma{g}=1 }
  \,\bk{A_0g, g}_{\ell^2}.
\end{equation*}
Let $g \in \cD(A_0)$ with $\norma{g}_{\ell^2} = 1$. 
 From \eqref{A1extim} we have
$$
  \norma{A_1(\xi) g}_{\ell^2}^2 \leq 2\abs{\xi}^2 \sum_n E_n \abs{g_n}^2 
  = 2\abs{\xi}^2 \bk{A_0g, g}_{\ell^2}
$$ 
and, therefore,
$\abs{\bk{A_1(\xi)g, g}_{\ell^2}}  \leq \norma{A_1(\xi)g}_{\ell^2} 
\leq \sqrt{2}\abs{\xi} \bk{A_0g, g}_{\ell^2}^{1/2}$, 
which, using $A(\xi) = A_0 + A_1(\xi) + A_2(\xi)$, yields 
\begin{equation}
\label{formext}
  \abs{ \bk{A(\xi)g, g}_{\ell^2} - \bk{A_0g, g}_{\ell^2} }
  \leq \sqrt{2}\abs{\xi}\,\bk{A_0g, g}_{\ell^2}^{1/2} + \frac{1}{2}\abs{\xi}^2.
\end{equation} 
 From \eqref{formext} we get, in particular,
$$
  \bk{A(\xi)g, g}_{\ell^2}  \leq  \bk{A_0g, g}_{\ell^2} 
  + \sqrt{2}\abs{\xi}\,\bk{A_0g, g}_{\ell^2}^{1/2} + \frac{1}{2}\abs{\xi}^2.
$$ 
which allows us to estimate $\lambda_n(\xi)$ from above.
In fact, since $x + \sqrt{2}\abs{\xi} x^{1/2} + \frac{1}{2}\abs{\xi}^2$ 
is an increasing function of $x$, we can write
$$
  \max\min \bk{A(\xi)g, g}_{\ell^2} \leq \max\min 
  \left\{\bk{A_0g, g}_{\ell^2} + \sqrt{2}\abs{\xi}\bk{A_0g, g}_{\ell^2}^{1/2} + \frac{1}{2}\abs{\xi}^2\right\}
$$ $$
  \leq \max\min \bk{A_0g, g}_{\ell^2} + \sqrt{2}\abs{\xi} \max\min  \bk{A_0g, g}_{\ell^2}^{1/2} 
  + \frac{1}{2}\abs{\xi}^2,
$$
that is 
\begin{equation}
\label{ext1}
  \lambda_n(\xi) \leq E_n + 2\abs{\xi}\,E_n^{1/2} + \frac{\abs{\xi}^2}{2},
\end{equation}
which holds for all $n \in \mN$.
We now estimate $\lambda_n(\xi)$ from below, at least for large $n$.
 From \eqref{formext} we get 
$$
  \bk{A(\xi)g, g}_{\ell^2} \geq  \bk{A_0g, g}_{\ell^2} - \sqrt{2}\abs{\xi}\bk{A_0g, g}_{\ell^2}^{1/2} 
  - \frac{1}{2}\abs{\xi}^2
$$ 
and we remark that $x - \sqrt{2}\abs{\xi} x^{1/2} - \abs{\xi}^2/2$ is an
increasing function of $x$ for $x \geq \abs{\xi}^2/2$.
Thus, let $n_0(\xi)$ be such that $E_{n_0(\xi)} \geq \abs{\xi}^2/2$
and fix $n \geq n_0(\xi)$.
Let us define 
$$
  S_{n-1}^0 = \mathrm{span} \{ e^{(1)}, e^{(2)},\ldots e^{(n-1)} \},
$$
where $\{ e^{(n)} \mid n \in \mN \}$ is the canonical basis of $\ell^2$ (eigenbasis of $A_0$).
We therefore have
\begin{equation*}
  \min_{g \in S_{n-1}^{0\perp} \cap \cD(A_0),\ \norma{g}=1 } \bk{A_0g,g}_{\ell^2} = E_n,
\end{equation*}
because $S_{n-1}^{0\perp} = \mathrm{span} \{ e^{(n)}, e^{(n+1)},\ldots \}$.
Thus, for every $g \in S_{n-1}^{0\perp} \cap \cD(A_0)$ with $\norma{g}_{\ell^2} = 1$, 
we can write 
$$
  \bk{A(\xi)g, g}_{\ell^2} \geq \bk{A_0g,g}_{\ell^2} - \sqrt{2}\abs{\xi}\,\bk{A_0g,g}_{\ell^2}^{1/2} 
  - \frac{1}{2}\abs{\xi}^2
  \geq E_n - \sqrt{2}\abs{\xi}\,E_n^{1/2} - \frac{1}{2}\abs{\xi}^2,
$$
(because $E_n \geq E_{n_0(\xi)} \geq \abs{\xi}^2/2$), and so
$$
  \min_{g \in S_{n-1}^{0\perp} \cap \cD(A_0),\ \norma{g}=1 } \bk{A(\xi)g, g}_{\ell^2} 
  \geq E_n - \sqrt{2}\abs{\xi}\,E_n^{1/2} - \frac{1}{2}\abs{\xi}^2.  
$$
Since $S_{n-1}^0 \in M_{n-1}$, we conclude that  
\begin{equation}
\label{ext2}
 \lambda_n(\xi) \geq E_n - \sqrt{2}\abs{\xi}\,E_n^{1/2} - \frac{1}{2}\abs{\xi}^2,
 \qquad n \geq n_0(\xi),
\end{equation}
which, together with \eqref{ext1}, yields \eqref{LambdaGrowth}.
\end{proof}
From \eqref{LambdaGrowth} we see that, for fixed $\xi$, the sequences $E_n$ and 
$\lambda_n(\xi)$ are asymptotically equivalent. 
Moreover it is not difficult to prove the following.
\begin{corollary}
\label{AnalyticRadius}
A constant $C_0$, independent of $n$, exists such that
$\lambda_{n}(\xi) \geq  \lambda_{n-1}(\xi)$ for all 
$\abs{\xi} \leq C_0(E_{n+1}-E_n)/\sqrt{E_n}$.
Then, the first $N$ bands do not cross each other in a ball of radius 
$$
 R_N = C_0 \max\{E_{n+1}-E_n \mid n \leq N+1 \} / \sqrt{E_{N+1}}\,.
$$
\end{corollary}
Let us now consider the family of diagonalization operators $\{ T(\xi) : \ell^2 \to \ell^2 \mid \xi \in \mR^d\}$, 
i.e.\ the unitary operators that map 1-1 the basis $\{e^{(n)} \mid n \in \mN\}$ onto the
basis $\{\varphi^{(n)}(\xi) | n \in \mN\}$, so that
\begin{equation}
\label{LambdaDef}
  \Lambda(\xi) = T^*(\xi) A(\xi) T(\xi) = 
\begin{pmatrix}
  \lambda_1(\xi) & 0 & 0 & \cdots
\\
  0 & \lambda_2(\xi) & 0 & \cdots
\\
   0 & 0 & \lambda_3(\xi) & \cdots
\\
 \vdots & \vdots & \vdots &\ddots  
\end{pmatrix}
\end{equation}
For any given $\eps \geq 0$ we define a unitary operator $T_\eps$ on the space $\cL^2$ 
(see Definition \ref{Spaces}) by
\begin{equation}
\label{Tdef}
  \big( T_\eps g \big)(k) = T(\eps k) g(k).
\end{equation}
\begin{theorem}
\label{P3}
For every $\eps \geq 0$, the operator $T_\eps : \cL^2 \to \cL^2$ 
is unitary, with $T_0 = I$.
Moreover, if $g \in {\cL^2_\mu}$ for some $\mu > 0$,
then $\lim_{\eps\to 0} \norma{T_\eps g - g}_{\cL^2} = 0$.
\end{theorem}
\begin{proof}
The first part of the statement is clear, because
$$
 \int_{\mR^d} \norma{T(\eps k) g(k)}^2_{\ell^2} \,dk 
 = \int_{\mR^d} \norma{g(k)}^2_{\ell^2} \,dk = \norma{g}_{\cL^2}^2
$$ 
and $\lambda_n(0) = E_n$.
Now, let $\Pi_N$ be the projection operator in $\ell^2$ on the $N$-dimensional sub-space
spanned by $e^{(1)}, e^{(2)}, \ldots, e^{(N)}$ (in other words, the cut-off operator after the 
$N$-th component).
Since the first $N$ bands do not cross in a ball of radius $R_N$ (see Corollary 
\ref{AnalyticRadius}), then $\xi \mapsto T(\xi)\Pi_N$ is unitary analytic from
$\Span \left\{e^{(1)}, e^{(2)}, \ldots, e^{(N)}\right\}$ to 
$\Span \left\{\varphi^{(1)}(\xi), \varphi^{(2)}(\xi), \ldots, \varphi^{(N)}(\xi)\right\}$,
in $\abs{\xi}\leq R_N$.
Let $g \in {\cL^2_\mu}$ and put
$$
  g^{(N)} = \Pi_Ng, \qquad g^{(N)}_c = g - g^{(N)},
$$
so that $\norma{(T_\eps -I) g}_{\cL^2} \leq \norma{(T_\eps -I) g^{(N)}}_{\cL^2}+\norma{(T_\eps -I) g^{(N)}_c}_{\cL^2}$.
Let $\eps>0$ and $r>0$ be such that $\eps r \leq R_N$.
Then, using the analyticity of $T(\eps k)\Pi_N$ in $\abs{\eps k}\leq \eps r \leq  R_N$, 
we can write
$$
  \norma{(T_\eps -I) g^{(N)}}_{\cL^2}^2
  = \int_{\mR^d} \norma{\left(T(\eps k) - I\right)g^{(N)}(k)}^2_{\ell^2}\,dk  
$$ $$
  = \int_{\abs{k} \leq r} \norma{\left(T(\eps k) - I\right) g^{(N)}(k) }^2_{\ell^2}\,dk
  + \int_{\abs{k} > r}\norma{\left(T(\eps k) - I\right) g^{(N)}(k)}^2_{\ell^2}\,dk 
$$ $$
  \leq L_N^2 \int_{\abs{k} \leq r} \abs{\eps k}^2 \norma{g^{(N)}(k)}^2_{\ell^2}\,dk
  + \frac{4}{r^{2\mu}} \int_{\abs{k} > r} \abs{\eps k}^{2\mu} \norma{g^{(N)}(k)}^2_{\ell^2}\,dk,
$$
for some Lipschitz constant $L_N > 0$.
Now, it can be easily verified that the inequality
\begin{equation}
\label{AuxIneq}
  \abs{k}^n \leq (1+\abs{k}^\mu)\,r^{\max\{ n-\mu,\,0\}}
\end{equation}
holds for any $r>0$, $n \geq 0$, $\mu \geq 0$, and $\abs{k} \leq r$.
 From this (with $n=1$) we get
$$
  \int_{\abs{k} \leq r} \abs{\eps k}^2 \norma{g^{(N)}(k)}^2_{\ell^2}\,dk
  \leq \eps^2 r^{2\max\{ 1-\mu,\,0\}} \int_{\abs{k} \leq r} 
  (1+\abs{\eps k}^\mu)^2 \norma{g^{(N)}(k)}^2_{\ell^2}\,dk
$$
and, therefore,
$$
 \norma{(T_\eps - I) g ^{(N)}}_{\cL^2}^2 \leq 
 \big(L_N^2\,\eps^2\, r^{2\max\{ 1-\mu,\,0\}}+ 4r^{-2\mu} \big)
 \norma{g}_{{\cL^2_\mu}}^2.
$$
Choosing $r = R_N/\eps$ we obtain
\begin{equation}
\label{TepsIneq}
  \norma{(T_\eps - I) g ^{(N)}}_{\cL^2} \leq \eps^{\min\{\mu,\,1\}}\,C(\mu,N)\,\norma{g}_{{\cL^2_\mu}},
 \end{equation}
where
$$
 C(\mu,N) = \left(L_N^2 R_N^{2\max\{ 1-\mu,\,0\}} + 4R_N^{-2\mu}\right)^{1/2}.
$$
Moreover, 
$$
 \norma{(T_\eps -I) g^{(N)}_c}_{\cL^2} \leq \norma{T_\eps g^{(N)}_c}_{\cL^2} + \norma{g^{(N)}_c}_{\cL^2}
 \leq 2\norma{g^{(N)}_c}_{\cL^2}.
$$ 
Since $\norma{g^{(N)}_c}_{\cL^2} \to 0$ as $N \to \infty$, we can fix $N$ and, then, $\eps$ in inequality 
\eqref{TepsIneq} so that $\norma{(T_\eps-I)g}_{\cL^2}$ is arbitrarily small, which proves the limit.
\end{proof}
\begin{remark}
 From inequality \eqref{TepsIneq} we see that, when a finite number $N$ of bands is considered, 
the distance between $T_\eps$ and $I$ is of order $\eps^{\min\{\mu,\,1\}}$ for $\gini \in {\cL^2_\mu}$,
with $\mu >0$.
\end{remark}
Let us now consider the second-order approximation of $\Lambda(\xi)$,
\begin{equation}
\label{Lambda2Def}
 \Lambda^{(2)}(\xi) =
\begin{pmatrix}
  \lambda_1^{(2)}(\xi) & 0 & 0 & \cdots
\\
  0 & \lambda_2^{(2)}(\xi) & 0 & \cdots
\\
   0 & 0 & \lambda_3^{(2)}(\xi) & \cdots
\\
 \vdots & \vdots & \vdots &\ddots  
\end{pmatrix}
\end{equation}
where $\lambda^{(2)}_n(\xi)$ is the second-order Taylor approximation of $\lambda_n(\xi)$:
\begin{equation*}
 \lambda_n(\xi) = \lambda^{(2)}_n(\xi) + \cO\big(\abs{\xi}^3\big).
\end{equation*}
The approximated eigenvalues $\lambda^{(2)}_n(\xi)$ can be computed by means of standard
non-degenerate perturbation techniques, which yield
\begin{equation}
\label{lambdaexp}
  \lambda^{(2)}_n(\xi) = E_n + \frac{1}{2}\, \xi \cdot \mM_n^{-1} \xi,
\end{equation}
where
\begin{equation}
\label{effm}
 \mM_n^{-1} = \nabla\otimes\nabla \,\lambda_n(\xi)_{\,|\xi=0} 
  = I - 2 \sum_{n'\not= n} \frac{P_{nn'} \otimes P_{n'n}}{E_n - E_{n'}}
\end{equation}
is the $n$-th band effective mass tensor \cite{Wenckebach99} (we remind that
$P_{nn'}=0$ if $n=n'$).
Note that the 1st order term in \eqref{lambdaexp} is zero.
\par
The operators $\Lambda(\xi)$ and $\Lambda^{(2)}(\xi)$, which
are self-adjoint on their maximal domains, generate,
respectively, the exact dynamics and the effective mass dynamics 
(in Fourier variables and in absence of external fields).
\begin{theorem}
\label{TheoEM}
Let $\gini \in {\cL^2_\mu}$, for some $\mu>0$, and assume $\gini = \Pi_N\gini$ (i.e.\ 
the initial datum is confined in the first $N$ bands).
Then, a constant $C(\mu,N,t) \geq 0$, independent of $\eps$, exists such that
\begin{equation}
\label{EMExtimN}
 \ds \norma{(\mathrm{e}^{-\frac{it}{\eps^2}\,\Lambda(\eps k)} -\mathrm{e}^{-\frac{it}{\eps^2}\,\Lambda^{(2)}(\eps k)}) \gini}_{\cL^2} 
 \leq \eps^{\min\{\mu/3,\, 1\}}\, C(\mu,N,t) \,\norma{\gini}_{{\cL^2_\mu}}\,.
\end{equation}
\end{theorem}
\begin{proof}
Note that, since $\Lambda(\eps k)$ and $\Lambda^{(2)}(\eps k)$ are diagonal, 
then both $\mathrm{e}^{-\frac{it}{\eps^2}\,\Lambda(\eps k)}\gini$  
and $\mathrm{e}^{-\frac{it}{\eps^2}\,\Lambda^{(2)}(\eps k)}\gini$ remain confined in the first $N$ bands at all times.
Denoting $g^\eps(t,k) = \mathrm{e}^{-\frac{it}{\eps^2}\,\Lambda^{(2)}(\eps k)}\gini$, the function 
$h^\eps(t,k) =  (\mathrm{e}^{-\frac{it}{\eps^2}\,\Lambda(\eps k)} -\mathrm{e}^{-\frac{it}{\eps^2}\,\Lambda^{(2)}(\eps k)}) \gini$ 
satisfies the Duhamel formula 
\begin{equation*}
   h^\eps( t,k) = \int_0^t \e^{- \frac{i(t-s)}{\eps^2}\,\Lambda(\eps k)}
 \,\frac{\Lambda(\eps k) - \Lambda^{(2)}(\eps k)}{\eps^2}\, g^\eps(s,k) \,ds,
\end{equation*}
so that
\begin{equation*}
\norma{ h^\eps(t,k)}_{\ell^2} \leq  \int_0^t \Big\|
  \frac{\Lambda(\eps k) - \Lambda^{(2)}(\eps k)}{\eps^2}\, g^\eps(s,k)\Big\|_{\ell^2} \,ds
\end{equation*}
Since $\lambda_1(\xi), \ldots \lambda_N(\xi)$ are analytic for $\abs{\xi} \leq R_N$
(see Corollary \ref{AnalyticRadius}), then
a Lipschitz constant $L'_N$ exists such that 
\begin{equation*}
 \Big\|\frac{\Lambda(\eps k) - \Lambda^{(2)}(\eps k)}{\eps^2}\, g^\eps(s,k)\Big\|_{\ell^2}
 \leq \eps L_N \abs{k}^3 \norma{g^\eps(k,s)}_{\ell^2}
 = \eps L_N \abs{k}^3 \norma{\gini(k)}_{\ell^2}
\end{equation*}
for all $k$ with $\abs{\eps k} \leq R_N$ (where we also used the fact that the $\ell^2$ norm
of $g^\eps$ is conserved during the unitary evolution).
Now we can proceed as in the proof of Theorem \ref{P3}:
if $r>0$ is such that $\eps r \leq R_N$, then we can write
$$
 \int_{\abs{k} \leq r}\norma{ h^\eps(t,k) }^2_{\ell^2} \, dk
\leq (L'_N t \eps)^2 \int_{\abs{k} \leq r} \abs{k}^6  \norma{\gini(k)}^2_{\ell^2}\,dk
$$
and, using inequality \eqref{AuxIneq} with $n=3$,
$$ 
 \int_{\abs{k} \leq r}\norma{ h^\eps(t,k) }^2_{\ell^2} \, dk
  \leq \left(L'_N t\,\eps\, r^{\max\{ 3-\mu,\, 0\}}\right)^2
  \norma{\gini}_{{\cL^2_\mu}}^2.
$$
Moreover,
$$
  \int_{\abs{k} > r}\norma{ h^\eps(t,k)}^2_{\ell^2} \, dk 
  \leq \frac{1}{r^{2\mu}} \int_{\abs{k} > r} \abs{k}^{2\mu} 
  \norma{h^\eps(t,k)}^2_{\ell^2}\,dk 
$$ $$
  \leq\frac{4}{r^{2\mu}} \int_{\abs{k} > r} \abs{k}^{2\mu} 
  \norma{\gini(t,k)}^2_{\ell^2}   \,dk
  \leq \frac{4}{r^{2\mu}} \norma{\gini}_{{\cL^2_\mu}}^2,
$$ 
where we used the fact that
$\norma{g^\eps(t,k)}_{\ell^2} = \norma{\gini(k)}_{\ell^2}$ for all $t$.
Hence, 
$$
 \norma{h^\eps(t)}_{\ell^2}^2 \leq 
 \left[ \left(L'_N t\,\eps\, r^{\max\{ 3-\mu,\, 0\}}\right)^2 +4r^{-2\mu} \right]
 \norma{\gini}_{{\cL^2_\mu}}^2
$$
and, choosing $r = R_N/\eps^{1/3}$, we obtain
$\norma{h^\eps(t)}_{\ell^2} \leq C(\mu,N,t)\, \eps^{\min\{\mu/3,\, 1\}} 
 \,\norma{\gini}_{{\cL^2_\mu}}$,
that is inequality \eqref{EMExtimN}, with
$$
 C(\mu,N,t) =  \left[ \left(L'_N t\, R_N^{\max\{ 3-\mu,\, 0\}}\right)^2 +4R_N^{-2\mu/3} \right]^{1/2}.
$$
\end{proof}
\begin{corollary}
Let  $\gini \in {\cL^2_\mu}$, with $\mu>0$ (but $\gini$ not necessarily confined in the first $N$ bands),  
then $\lim_{\eps\to 0} \norma{(\mathrm{e}^{-\frac{it}{\eps^2}\,\Lambda(\eps k)} -\mathrm{e}^{-\frac{it}{\eps^2}\,\Lambda^{(2)}(\eps k)}) \gini}_{\cL^2}  = 0$,
uniformly in bounded  time intervals.
\end{corollary}
\begin{proof}
Like in the proof of the above theorem, we define 
\begin{equation*}
   h^\eps( t,k) = \int_0^t \e^{- \frac{i(t-s)}{\eps^2}\,\Lambda(\eps k)}
 \,\frac{\Lambda(\eps k) - \Lambda^{(2)}(\eps k)}{\eps^2}\, g^\eps( s,k) \,ds,
\end{equation*}
For any given $N$ we can write
$$
  \norma{h^\eps(t)}_{\cL^2} \leq \norma{\Pi_N h^\eps(t)}_{\cL^2} + \norma{\Pi_N^c h^\eps(t)}_{\cL^2},
$$
where $\Pi_N^c = I - \Pi_N$.
Recalling that the evolutions are diagonal, the first term at the right hand side 
corresponds to the initial datum $\Pi_N\gini$, for which \eqref{EMExtimN} holds. 
Using the fact that $\Pi_N$ commutes with both $\e^{-\frac{it}{\eps^2}\,\Lambda(\eps k)}$
and $\e^{-\frac{it}{\eps^2}\,\Lambda^{(2)}(\eps k)}$, for the second term we have
$$
\norma{\Pi_N^c h^\eps(t)}_{\cL^2}  
\leq 2\norma{\Pi_N^c \gini}_{\cL^2}.
$$
Since $\Pi_N^c \gini \to 0$ in $\cL^2$ as $N\to\infty$, this inequality, 
together with \eqref{EMExtimN}, shows that 
$\lim_{\eps \to 0}\norma{h^\eps(t)}_{\cL^2} = 0$, uniformly in bounded $t$-intervals.
\end{proof}
\section{Comparison of the models}
\label{sec5}
We are now in position to exhibit the ensemble of models encountered and to compare their respective dynamics.
\par
We first started by the exact dynamics. 
Let the wave function $\psi^\eps(t,x)$ be solution of the initial value problem \eqref{SE1}.
If we denote by $f_n^{in,\eps}(x)$ the $\eps$-scaled envelope 
functions of the initial wave function  $\psi^{in,\eps}$, relative to the basis $v_n$,  and by $g_n^{in,\eps}(k)$,
their Fourier transform, then the Fourier transformed envelope functions $g^\eps$ of $\psi^\eps(t,x)$ are 
the solutions of 
\begin{align}
\label{EFEexact}
  &i\partial_t \,g = A_\KP^\eps g + \cU^\eps g, & g^\eps(t=0) = g^{in,\eps}\quad &\text{(exact dynamics)}
\end{align}
where 
\begin{equation}
\label{HKPdef}
 \left( A_\KP^\eps g\right)_n(k) = \frac{1}{\eps^2} \left(A(\eps k)\,g(k)\right)_n
 =  \left( \frac{E_n}{\eps^2} + \frac{\abs{k}^2}{2} \right)g_n(k)
  - \frac{i}{\eps} \sum_{n'} k \cdot P_{nn'} g_{n'}(k),
\end{equation}
The k$\cdot$p approximation consists in passing to the limit in $\cU^\eps$. 
Therefore, we define $ g^\eps_\KP(t)$ as the solution of
\begin{align}
\label{EFEkp}
  &i\partial_t \,g = A_\KP^\eps g + \cU^0 g,&g(t=0) = g^{in,\eps} &\quad \text{(k$\cdot$p model)}
\end{align}
It is worth noting that the back Fourier transform of $g^\eps_\KP(t)$ which we will denote by 
$f^\eps_\KP(t,x)$
is a solution of system
\begin{equation}
\label{kpF}
\begin{aligned}
 &i\partial_t \,f_n =  {E_n\over \eps^2} f_{n}
  - \frac{1}{2} \Delta\,f_n
 - \frac{1}{\eps} \sum_{n'} P_{nn'}\cdot\nabla f_{n'}
  + \sum_{n'}  V_{n n'} f_{n'}, 
\\
 &f_n(t=0) = f_n^{in,\eps}(x).
\end{aligned}
\end{equation}
The diagonalization of the operator $A_\KP^\eps$ performed in the previous section leads to the 
effective mass dynamics
\begin{align}
\label{EFEem}
  &i\partial_t \,g = A^\eps_\EM g + \cU^0 g,& g(t=0) = g^{in,\eps} \quad &\text{(effective mass model)}
\end{align}
where \begin{equation}
\label{HEMdef}
 \left( A^\eps_\EM g\right)(k) = \frac{1}{\eps^2}\left(\Lambda^{(2)}(\eps k)\,g(k)\right)_n
 = \left(\frac{E_n}{\eps^2} + \frac{1}{2}\, k \cdot \mM_n^{-1} k \right) g_n(k).
\end{equation}
The solution of \eqref{EFEem} will be denoted by $g^\eps_\EM(t,k)$ and its back Fourier transform 
$f^\eps_\EM(t,x)$ is easily shown to be the solution of 
\begin{equation}
\label{EM1pos}
\begin{aligned}
  &i\partial_t\, f_n =  \frac{1}{\eps^2}\,E_n f_{n}
  - \frac{1}{2} \DIV \left( \mM_n^{-1} \nabla f_{n}^\eps \right)
  + \sum_{n'} V_{nn'}\,f_{n'}^\eps, 
\\
&f_n(t=0) = f_n^{in,\eps}(x).
\end{aligned}
\end{equation}
This equation is still involving oscillations in time. These oscillations can be filtered by setting
$f_{n,\EM}^\eps (t,x) = h_{n,\EM}^\eps(t,x) \mathrm{e}^{-iE_n{t\over \eps^2}}$  which will be a solution of 
\begin{equation}
\label{hdyn}
\begin{aligned}
  &i\partial_t\, h_{\EM,n}^\eps =   
  - \frac{1}{2} \DIV \left( \mM_n^{-1} \nabla h_{\EM,n}^\eps \right)
  + \sum_{n'} \e^{i\omega_{nn'}t/\eps^2} V_{nn'}\,h_{\EM,n'}^\eps, 
\\
&h_{\EM,n}^\eps(t=0) = f_n^{in,\eps}(x),
\end{aligned}
\end{equation}
where
\begin{equation}
  \omega_{nn'} = E_n - E_{n'}.
\end{equation}
The limit $h_{\EM,n}$ of these function is the solution of the system
\begin{equation}
\label{limit}
  i\partial_t\, h_{\EM,n} =   
  - \frac{1}{2} \DIV \left( \mM_n^{-1} \nabla h_{\EM,n} \right)
 +V_{nn}\,h_{\EM,n}, \quad h_{\EM,n}(t=0) = f_n^{in}(x),
\end{equation}
where $f_n^{in}(x)$ is the limit as $\eps$ tends to zero of $f_n^{in,\eps}(x)$, and which will be made precise later on.
\begin{remark}
The external-potential operators $\cU^\eps$ and $\cU^0$ have been defined in \eqref{Udef} and \eqref{U0def}.
The free k$\cdot$p operator $A(\xi)$ and the effective mass operator 
$\Lambda^{(2)}(\xi)$ (see definitions \eqref{Adef} and \eqref{Lambda2Def}) are now re-introduced as operators 
acting in $\cL^2$.
Recalling definition \eqref{LambdaDef}, we shall also consider the diagonal k$\cdot$p operator
\begin{equation}
\label{Lambda}
 \left( \Lambda^\eps g\right)_n(k) = \frac{1}{\eps^2} \left(\Lambda(\eps k)\,g(k)\right)_n
 = \frac{1}{\eps^2}\,\lambda_n(\eps k)\,g_n(k).
\end{equation}
The operators $A_\KP^\eps$, $A^\eps_\EM$ and $\Lambda^\eps$ are ``fibered'' self-adjoint operators in $\cL^2$,
with fiber space $\ell^2$.
It is well known (see Ref.~\cite{ReedSimonIV78}) that a fibered self-adjoint operator $L$ in $\cL^2$ 
has self-adjointness domain
$$
\cD(L) = \Big\{ g \in \cL^2 \;\Big|\; g(\xi) \in \cD\left(L(\xi)\right) \text{\ a.e. $\xi \in \mR^d$} 
 \text{\ and $\int_{\mR^d} \norma{L(\xi)\,g(\xi)}^2_{\ell^2}\, d\xi < \infty$} \Big\},
$$
where $\cD\left(L(\xi)\right)$ is the self-adjointness domain of $L(\xi)$ in $\ell^2$.
\end{remark}
\subsection{Comparison of Envelope functions}
Assuming $V \in \cW_0$ (Definition \ref{Spaces}), we know from Lemma \ref{Lemma0} 
that $\cU^\eps$ and $\cU^0$ are bounded (and, clearly, symmetric). 
Therefore, $A_\KP^\eps + \cU^\eps$, $A_\KP^\eps + \cU^0$ and $A_\EM^\eps + \cU^\eps$
are the generators of the unitary evolution groups
$$
  G^\eps(t) = \e^{-it( A_\KP^\eps + \cU^\eps)}, 
\quad
  G^\eps_\KP(t) = \e^{-it( A_\KP^\eps + \cU^0)}, 
\quad
  G^\eps_\EM(t) = \e^{-it( A_\EM^\eps + \cU^0)}.
$$
Our goal is to compare, in the limit of small $\eps$, the three mild solutions 
of Eqs.~\eqref{EFEexact}, \eqref{EFEkp} and \eqref{EFEem}, i.e.
\begin{equation}
\label{MildSols}
  g^\eps(t) = G^\eps(t)\,g^{in,\eps},
\quad
  g^\eps_\KP(t) = G^\eps_\KP(t)\,g^{in,\eps},
\quad
  g^\eps_\EM(t) = G^\eps_\EM(t)\,g^{in,\eps},
\end{equation}
\begin{lemma}
\label{L3}
Let $g^{in,\eps} \in {\cL^2_\mu}$ and $V \in \cW_\mu$ for some $\mu \geq 0$ (see Definition \ref{Spaces}).
Then, suitable constants $c_1(\mu,V) \geq 0$ and  $c_2(\mu,V) \geq 0$, independent of $\eps$, exists such that 
\begin{equation}
\label{L3ineq}
  \norma{g^\eps_\KP(t)}_{{\cL^2_\mu}} \leq 
  \e^{c_1(\mu,V)t} \,\norma{g^{in,\eps}}_{{\cL^2_\mu}},
\qquad
  \norma{g^\eps_\EM(t)}_{{\cL^2_\mu}} \leq 
  \e^{c_2(\mu,V)t} \,\norma{g^{in,\eps}}_{{\cL^2_\mu}},
\end{equation}
for all $t \geq 0$.
\end{lemma}
\begin{proof}
We prove the lemma only for $g_\KP$, the proof for $g_\EM$ being identical. 
We also skip the $\eps$ superscript of $g^{in,\eps}$.
Let $\alpha$ be a fixed multi-index with $\abs{\alpha} \leq \mu$.
For $R>0$, consider the bounded multiplication operators on $\cL^2$
$$
  \big( m_R\, g \big)_n(k) = \left\{ 
            \begin{aligned} 
            &k^\alpha g_n(k), & &\text{if $\abs{k}\leq R$,} 
         \\ 
            &0, & &\text{otherwise,} 
            \end{aligned}
            \right.
$$
Moreover, we denote by $m_\infty$ the (unbounded) limit operator
$\big( m_\infty\, g \big)_n(k)  = k^\alpha g_n(k)$.
Since $m_R$ (with $R < \infty$) commutes with $A_\KP^\eps$ on $\cD(A_\KP^\eps)$,
then, by applying standard semigroup techniques, we obtain
\begin{equation*}
  m_R\, g^\eps_\KP(t) = G^\eps_\KP(t)\, m_R\, \gini + \int_0^t G^\eps_\KP(t-s) 
  \left[ m_R ,\,\cU^0\right] g^\eps_\KP(s) \,ds
\end{equation*}
and, therefore,
\begin{equation}
\label{L3aux1}
 \norma{ m_R\, g^\eps_\KP(t)}_{\cL^2} \leq \norma{ m_R\, \gini}_{\cL^2} 
 + \int_0^t \norma{ \left[ m_R,\, \cU^0\right] g^\eps_\KP(s)}_{\cL^2} \,ds.
\end{equation}
Using \eqref{U0def} and the identity $k^\alpha - \eta^\alpha = \sum_{\beta < \alpha} \binom{\alpha}{\beta} 
 \,( k - \eta)^{\alpha-\beta}\,\eta^\beta$, we have
$$
 \big( \left[m_\infty,\, \cU^0 \right] g^\eps_\KP\big)_n(k)
 = \sum_{n'} \frac{1}{(2\pi)^{d/2}} \int_{\mR^d} 
  (k^\alpha - \eta^\alpha)\,\hat V_{nn'}(k-\eta)\,g^\eps_{\KP, n'}(\eta)\,d\eta
$$ $$
 = \sum_{n'} \frac{1}{(2\pi)^{d/2}} \sum_{\beta < \alpha} \binom{\alpha}{\beta} 
  \int_{\mR^d} (k - \eta)^{\alpha-\beta}\,\hat V_{nn'}(k-\eta)\,\eta^\beta g^\eps_{\KP, n'}(\eta)\,d\eta
$$
Since $V \in \cW_\mu$, the potential $U_{\alpha\beta}(x,z)$ such that 
$\hat U_{\alpha\beta}(k,z) = k^{\alpha - \beta}\hat V(k,z)$
belongs to $\cW_0$, with $\norma{U_{\alpha\beta}}_{\cW_0} \leq \norma{V}_{\cW_\mu}$, and then,
using \eqref{Knorm}, we obtain 
\begin{equation}
\label{Comm1}
 \norma{\left[m_\infty,\, \cU^0 \right] g^\eps_\KP}_{\cL^2} \leq 
 \sum_{\beta < \alpha} \binom{\alpha}{\beta} \norma{U_{\alpha\beta}}_{\cW_0}
 \norma{\eta^\beta g^\eps_\KP}_{\cL^2}
 \leq c_1(\mu,V) \norma{g^\eps_\KP}_{{\cL^2_\mu}}.
\end{equation}
with $c_1(\mu,V) = (2^d-1)\norma{V}_{\cW_\mu}$.
Letting $R\to+\infty$, it is not difficult to show that the dominated convergence theorem applies 
and yields
$$
  \lim_{R \to +\infty} \norma{\left[m_R,\, \cU^0 \right] g^\eps_\KP}_{\cL^2}
  = \norma{\left[m_\infty,\, \cU^0 \right] g^\eps_\KP}_{\cL^2} 
  \leq c_1(\mu,V) \norma{g^\eps_\KP}_{{\cL^2_\mu}}.
$$
Then, passing to the limit for $R \to +\infty$ in \eqref{L3aux1}, we get 
\begin{equation*}
  \norma{g^\eps_\KP(t)}_{{\cL^2_\mu}} \leq \norma{\gini}_{{\cL^2_\mu}}
 + c_1(\mu,V)\int_0^t \norma{g^\eps_\KP(s)}_{{\cL^2_\mu}} \,ds,
\end{equation*}
and, therefore, Gronwall's Lemma yields inequality \eqref{L3ineq}.
\end{proof}
Let us begin by comparing the exact dynamics $g^\eps(t)$ with the 
k$\cdot$p dynamics $g_\KP^\eps(t)$.
\begin{theorem}
\label{EXvsKP}
Let $g^\eps(t)$ and $g^\eps_\KP(t)$ be respectively the solution of  \eqref{EFEexact} 
and \eqref{EFEkp}.
If $\ginieps \in {\cL^2_\mu}$ and $V \in \cW_\mu$, for some $\mu \geq 0$, 
then, for any given $\tau \geq 0$, a constant $C(\mu,V,\tau) \geq 0$,
independent of $\eps$, exists such that 
\begin{equation}
\label{EXvsKPeq}
  \norma{g^\eps(t) - g^\eps_\KP(t)}_{\cL^2} 
  \leq \eps^\mu\, C(\mu,V,\tau)\, \norma{\gini}_{{\cL^2_\mu}},
\end{equation} 
for all $0 \leq t \leq \tau$.
\end{theorem}
\begin{proof}
The function $h^\eps(t) = g^\eps(t) - g^\eps_\KP(t)$ satisfies the integral equation
\begin{equation*}
  h^\eps(t) = \int_0^t G^\eps(t-s) \big(\cU^0 - \cU^\eps \big) g^\eps_\KP(s) \,ds
\end{equation*}
and, therefore,
\begin{equation*}
 \norma{h^\eps(t)}_{\cL^2} \leq  \int_0^t \norma{\big( \cU^0 - \cU^\eps \big)
 g^\eps_\KP(s)}_{\cL^2} \,ds.
\end{equation*}
 From Lemma \ref{L3} we have that $g^\eps_\KP(t)$ belongs to $\cL^2_\mu$
for all $t$ and, therefore, we can apply Theorem \ref{T3}, which gives
\begin{equation*}
  \norma{\big( \cU^0 - \cU^\eps \big) g^\eps_\KP(s)}_{\cL^2} \leq 
  \eps^\mu c_\mu \norma{V}_{\cW_\mu}\,\norma{g^\eps_\KP(s)}_{{\cL^2_\mu}},
\end{equation*}
for a suitable constant $c_\mu$.
Then we have
\begin{equation*}
  \norma{h^\eps(t)}_{\cL^2}  \leq \eps^\mu\,c_\mu\,\norma{V}_{\cW_\mu}\,
  \int_0^t \norma{g^\eps_\KP(s)}_{{\cL^2_\mu}}ds
\end{equation*}
and, by \eqref{L3ineq}, we have that \eqref{EXvsKPeq} holds with
$ C(\mu,V,\tau) = 
  \frac{c_\mu \left(\e^{c(\mu,V) \tau} - 1\right)}{c(\mu,V)} \norma{V}_{\cW_\mu}$.
\end{proof}
We now compare the k$\cdot$p dynamics $g_\KP^\eps(t)$ with the effective mass 
dynamics $g_\EM^\eps(t)$ (see definitions \eqref{MildSols}).
Recalling the discussion in Sec.~\ref{sec4}, we need, as an intermediate step between
$g_\KP^\eps(t)$ and $g_\EM^\eps(t)$, the function $g_*^\eps(t) = T^*_\eps\,g^\eps_\KP(t)$, 
that is
\begin{equation}
\label{GD}
  g_*^\eps(t) = T^*_\eps\,G^\eps_\KP(t)\,\ginieps= 
  \exp\left[-it\,\left( \Lambda^\eps + T^*_\eps\,\cU^0 T_\eps \right)\right]\,T^*_\eps\,\ginieps,
\end{equation}
representing the diagonalized k$\cdot$p dynamics 
(definitions \eqref{LambdaDef}, \eqref{Tdef} and \eqref{Lambda}).
\begin{lemma}
\label{KPvsEM}
Let $g^\eps_\EM(t)$ and $g_*^\eps(t)$ be respectively defined by  \eqref{EFEem} and \eqref{GD}.
Let $\ginieps \in {\cL^2_\mu}$ and $V \in \cW_\mu$, for some $\mu > 0$, and assume 
$\ginieps = \Pi_N\gini$ (i.e.\ $\ginieps$ is concentrated in the first $N$ bands). 
Then, for any given $\tau \geq 0$, a suitable constant $C'(\mu,N,V,\tau)$, independent of $\eps$, 
exists such that
\begin{equation}
\label{LastIneq}
  \norma{g_*^\eps(t) - g^\eps_\EM(t)}_{\cL^2} 
  \leq \eps^{\min\{\mu/3,\,1 \}}\,C'(\mu,N,V,\tau)\,\norma{\ginieps}_{{\cL^2_\mu}},
\end{equation}
for all $0 \leq t \leq \tau$.
\end{lemma}
\begin{proof}
Let $S^\eps_{\Lambda}(t) = \exp(-it \Lambda^\eps)$, $S^\eps_\EM(t) = \exp(-it A^\eps_\EM)$ 
and $\cU^\eps_T : = T^*_\eps\,\cU^0 T_\eps$.
Then, 
$$
\begin{aligned}
 &g^\eps_\EM(t) = S^\eps_\EM(t) \ginieps + \int_0^t S^\eps_\EM(t-s)\, 
 \cU^0 g^\eps_\EM(s)\,ds,
\\[4pt]
 &g_*^\eps(t) = S^\eps_\Lambda(t) T_\eps^* \ginieps + \int_0^t S^\eps_\Lambda(t-s)\,
 \cU^\eps_T\,g_*^\eps(s)\,ds.
\end{aligned}
$$
Putting $h^\eps = g_*^\eps - g^\eps_\EM$, we can write
\begin{multline}
\label{Duhamel}
 h^\eps(t) = \left( S^\eps_\Lambda- S^\eps_\EM\right)(t)\, \ginieps 
 + S^\eps_\Lambda(t) \left(T_\eps^*  - I\right)\ginieps
 + \int_0^t S^\eps_\Lambda(t-s)\,\cU^\eps_T h^\eps(s)\,ds
\\
 + \int_0^t S^\eps_\Lambda(t-s) \left( \cU^\eps_T - \cU^0 \right) g^\eps_\EM(s)\, ds
 + \int_0^t \left( S^\eps_\Lambda - S^\eps_\EM\right)(t-s)\, \cU^0 g^\eps_\EM(s)\,ds.
\end{multline}
From the effective mass theorem, Theorem \ref{TheoEM},
a constant  $C(\mu,N,t)$ exists such that 
\begin{equation}
\label{Stima1}
  \norma{ \left( S^\eps_\Lambda- S^\eps_\EM\right)(t)\, \ginieps}_{\cL^2}
  \leq \eps^{\min\{\mu/3,\, 1\}}\,C(\mu,N,t) \norma{\ginieps}_{{\cL^2_\mu}}.
\end{equation}
Moreover, from Lemma \ref{L3} we have that both $g^\eps_\EM(t)$ and $\cU^0 g^\eps_\EM(t)$ belong 
to $\cL^2_\mu$ for all $t$, and that a constant 
$C_1(\mu,V,t) \geq 0$ exists such that
\begin{equation}
\label{Stima1.1}
  \norma{\cU^0 g^\eps_\EM(t)}_{{\cL^2_\mu}} \leq C_1(\mu,V,t)\, \norma{\ginieps}_{{\cL^2_\mu}}
\end{equation}
(this stems, in particular, from the commutator inequality \eqref{Comm1}, which still holds for $g^\eps_\EM$).
This inequality, together with Theorem \ref{TheoEM}, yields
\begin{equation}
\label{Stima2}
  \norma{\left( S^\eps_\Lambda- S^\eps_\EM\right)(t-s)\, \cU^0 g^\eps_\EM(s)}_{\cL^2}
  \leq \eps^{\min\{\mu/3,\, 1\}}\,C_2(\mu,N,t-s) \norma{\ginieps}_{{\cL^2_\mu}}
\end{equation}
for a suitable constant $C_2(\mu,N,t) \geq 0$.
In order to estimate the last integral in \eqref{Duhamel}, let us write
$$
  \big(\cU^\eps_T - \cU^0 \big) g^\eps_\EM(s) 
  = \left(T_\eps^* - I\right) \cU^0 g^\eps_\EM(s) 
  + T_\eps^*\,\cU^0 \left(T_\eps - I \right) g^\eps_\EM(s).
$$
Using inequalities \eqref{TepsIneq} and \eqref{Stima1.1} we see that another constant
$C_3(\mu,N,V,t) \geq 0$ exists such that
\begin{equation}
\label{Stima3}
 \norma{\big( \cU^\eps_T - \cU^0 \big) g^\eps_\EM(s)}_{\cL^2} 
 \leq \eps^{\min\{\mu,1\}}\, C_3(\mu,N,V,t)\,\norma{\ginieps}_{{\cL^2_\mu}} \,.
\end{equation}
In conclusion, from inequalities \eqref{Stima1}, \eqref{Stima2} and \eqref{Stima3},
and from Eq.~\eqref{Duhamel}, we get 
$$
  \norma{h^\eps(t)}_{\cL^2} \leq  \eps^{\min\{\mu/3, 1\}} \,C_4(\mu,N,V,\tau) \norma{\ginieps}_{{\cL^2_\mu}}
   + \norma{V}_{\cW_0}\int_0^t \norma{h^\eps(s)}_{\cL^2}\,ds,
$$
for all  $0 \leq t \leq \tau$ (here we also used the fact that all the estimation constants
introduced so far are non-decreasing with respect to time).
Hence, inequality \eqref{LastIneq}, with $C'(\mu,N,V,\tau) = \e^{\tau\norma{V}_{\cW_0}} C_4(\mu,N,\tau,V)$, 
follows from Gronwall's Lemma.
\end{proof}
\begin{theorem}
\label{CoroFinal}
Let  $g^\eps_\KP(t)$ and $g_\EM^\eps(t)$ as in \eqref{MildSols}, and 
assume $g^{in,\eps} \in {\cL^2_\mu}$, with a uniform bound as $\eps$ tends to zero.
Moreover, assume $V \in \cW_\mu$ for some $\mu > 0$.
Then $\lim_{\eps\to 0} \norma{g^\eps_\KP(t) - g^\eps_\EM(t)}_{\cL^2} =0$, uniformly
in bounded time-intervals.
If, in addition, $g^{in,\eps} = \Pi_Ng^{in,\eps}$, for some $N$ then, for any given $\tau \geq 0$, 
a constant $C''(\mu,N,V,\tau) \geq 0$, independent of $\eps$, exists such that
\begin{equation}
\label{CoroIneq}
  \norma{g^\eps_\KP(t) - g^\eps_\EM(t)}_{\cL^2} 
  \leq \eps^{\min\{\mu/3,\,1 \}}\,C''(\mu,N,V,\tau)\,\norma{g^{in,\eps}}_{{\cL^2_\mu}},
\end{equation}
for all $0 \leq t \leq \tau$.
\end{theorem}
\begin{proof}
We begin by the second statement, assuming $g^{in,\eps} = \Pi_Ng^{in,\eps}$.
Using inequalities \eqref{TepsIneq} and \eqref{LastIneq}, and recalling definition \eqref{GD}, 
we can write
$$
  \norma{g^\eps_\KP(t) - g^\eps_\EM(t)}_{\cL^2} \leq 
  \norma{(T_\eps^*-I)g^\eps_\KP(t)}_{\cL^2} + \norma{g^\eps_*(t) - g^\eps_\EM(t)}_{\cL^2}
$$ $$
  \leq \eps^{\min\{\mu,\,1\}}\,C(\mu,N)\,\norma{g^\eps_\KP(t)}_{{\cL^2_\mu}} 
  + \eps^{\min\{\mu/3,\,1 \}}\,C'(\mu,N,V,\tau)\,\norma{g^{in,\eps} }_{{\cL^2_\mu}},
$$
for $0 \leq t \leq \tau$.
Then, using also \eqref{L3ineq}, inequality \eqref{CoroIneq} follows.
If now $g^{in,\eps}$ simply belongs to $\cL^2_\mu$, then for any fixed $N$ we can write
$$
  \norma{(g^\eps_\KP - g^\eps_\EM)(t)}_{\cL^2}  \leq 
  \norma{(G^\eps_\KP - G^\eps_\EM)(t)\Pi_N g^{in,\eps} }_{\cL^2} 
  + \norma{(G^\eps_\KP - G^\eps_\EM)(t)\Pi_N^c g^{in,\eps} }_{\cL^2}
$$ $$
 \leq \eps^{\min\{\mu/3,\,1 \}}\,C''(\mu,N,\tau)\,\norma{g^{in,\eps}}_{{\cL^2_\mu}}
 + 2\,\norma{\Pi_N^c g^{in,\eps} }_{\cL^2},
$$
for all $0 \leq t \leq \tau$.
Since $\Pi_N^c g^{in,\eps}  \to 0$ in $\cL^2$ as $N\to\infty$, then we can fix $N$ large enough and, 
successively, $\eps$ small enough (uniformly in $0 \leq t \leq \tau$, by assumption) so that
$\norma{g^\eps_\KP(t) - g^\eps_\EM(t)}_{\cL^2}$ is arbitrarily small, which proves our assertion.
\end{proof}
The following result is a direct consequence of the above comparisons.
\begin{corollary}
Assume that the envelope functions $f^{in,\eps}(x)$ are bounded in $\cH^\mu$ for some $\mu>0$ 
and that the potential $V(x,z)$ belong to $\cW_\mu$. 
Then we have the local uniform in time convergence
$$
  \lim_{\eps \to 0} \norma{f^\eps(t) - f^\eps_\EM(t)}_{\cL^2} = 0,
$$
where $f^\eps(t,x)$ and $f^\eps_\EM(t,x)$ are the respective solutions of \eqref{EXE} and \eqref{EM1pos}
\end{corollary}
We are now able to prove the following theorem.
\begin{theorem}
\label{LastTeo}
Let $h_\EM^\eps(t,x)$ and $h_\EM(t,x)$ be the mild solutions of, respectively,  Eq.~\eqref{hdyn} 
and Eq.~\eqref{limit}. 
Assume $\lim_{\eps\to0} \norma{f^{in,\eps} - f^{in}}_{\cL^2} = 0$ and assume that $\mu > 0$  exists
such that $V \in \cW_\mu$ and $f^{in,\eps}$ is bounded uniformly in $\cH^\mu$.
Then 
$$
  \lim_{\eps\to0} \norma{h_\EM^\eps(t) - h_\EM(t)}_{\cL^2} = 0,
$$
uniformly in bounded time intervals.
\end{theorem}
\begin{proof}
Since the dynamics generated by \eqref{hdyn} and \eqref{limit} both preserve the $\cL^2$ norm, 
we can assume without loss of generality that the initial condition $f^{in,\eps}$ and $f^{in}$ 
are identical and replace them by the notation $\hini \in \cH^\mu$. 
We consider the diagonal operator $H_0$ in $\cL^2$
$$
  \left(H_0 h\right)_n(x) = 
  \frac{1}{2} \DIV \left( \mM_n^{-1} \nabla h_{n} \right)(x) + V_{nn}(x)h_{n}(x).
$$
We recall that the matrix $V_{nn'}$ defines a bounded operator on $\cL^2$ (that is,
the operator $\cU^0$ in position variables, see definition \eqref{U0def}). 
Such operator, as well as its diagonal and off-diagonal parts are bounded operators 
with bound $\norma{V}_{\cW_0}$ (see Lemma \ref{Lemma0}).
Then, $H_0$ is self-adjoint on the domain
$$
  \cD(H_0) = \left\{ h \in \cL^2 \;\left|\;  h_n \in \rH^2(\mR^d),\ 
  \sum_{n} \norma{\DIV(\mM_n^{-1} \nabla h_{n})}_{L^2(\mR^d)}^2 < \infty \right. \right\}.
$$
Let $S(t) = \exp(-itH_0)$ denote the (diagonal) unitary group generated by $H_0$.
Moreover we consider the operator $R^\eps(t)$ given by
$$
  \left(R^\eps(t)h\right)_n(x) = 
  \sum_{n' \not= n} \e^{i\omega_{nn'}t/\eps^2} V_{nn'}(x)\,h_{n'}(x),
$$
which, being unitarily equivalent to the off-diagonal part of $\cU^0$, is again 
bounded by $\norma{V}_{\cW_0}$ (for all $t$).
The two mild solutions satisfy
$$
 h_\EM^\eps(t) = S(t) \hini + \int_0^t S(t-s) R^\eps(s)\, h_\EM^\eps(s)\,ds,
\qquad
 h_\EM(t) = S(t) \hini,
$$
and, therefore, what we need to do is proving that 
$$
 h_\EM^\eps(t) - h_\EM(t) = \int_0^t S(t-s) R^\eps(s)\, h_\EM^\eps(s)\,ds
$$ 
goes to zero as $\eps\to 0$.
To this aim we resort to the usual cutoff. 
For any fixed $N \in \mN$ we decompose the right hand side of the previous equation
\begin{equation*}
   h_\EM^\eps(t) - h_\EM(t) = I_N(t) + I_N^c (t),
\end{equation*}
where, using the projection operators $\Pi_N$ and $\Pi_N^c = I - \Pi_N$,
introduced in the proof of Theorem \ref{P3}, we have put
$$
\begin{aligned}
 &I_N(t) = \int_0^t S(t-s) \Pi_N R^\eps(s) \Pi_N h_\EM^\eps(s)\,ds,
\\
 &I_N^c(t) = \int_0^t S(t-s) 
  \left[\Pi_N^c R^\eps(s) \Pi_N  + \Pi_NR^\eps(s) \Pi_N^c 
  + \Pi_N^c R^\eps(s) \Pi_N^c\right] h_\EM^\eps(s)\,ds.
\end{aligned}
$$
\paragraph{Case of regular data}
We assume in this part that $V\in \cW_2$ and that $\hini \in \cH^2$.
We fix a $\delta>0$ arbitrarily small and a maximum time $\tau$.
Because $R^\eps(t)$ is uniformly bounded  and 
$\norma{h_\EM^\eps(t)}_{\cL^2} = \norma{\hini}_{\cL^2}$ then, clearly, a number $N(\delta,\tau)$
(independent of $\eps$) exists such that
$\norma{I_N^c(t)}_{\cL^2} \leq \delta$, for all $N \geq N(\delta,\tau)$ and $0 \leq t \leq \tau$.
We now turn our attention to $I_N(t)$.
Using the assumption $V \in \cW_2$, it is not difficult to prove 
the following facts:
\begin{enumerate}
\item[(i)]
 for every $N$, if $h \in {\cH^2}$ then $\Pi_N h \in \cD(H_0)$, and a constant 
 $C_N$ exists such that $\norma{H_0\,\Pi_Nh}_{\cL^2} \leq C_N\norma{h}_{{\cH^2}}$;
\item[(ii)]
 for every $N$ a constant $C_N'$, independent of $t$ and $\eps$, exists such that, 
 if $h \in {\cH^2}$, then $\norma{\Pi_NR^\eps(t)\Pi_Nh}_{{\cH^2}} \leq C_N'\norma{h}_{{\cH^2}}$\,.
\end{enumerate}
Moreover, in a similar way to Lemma \ref{L3}, we can prove the following: 
\begin{enumerate}
\item[(iii)]
 if $\hini \in {\cH^2}$, then $h_\EM^\eps(t) \in {\cH^2}$ for all $t$ and a function
 $C(t)$, bounded on bounded time intervals and independent of $\eps$, exists such that
 $\norma{h_\EM^\eps(t)}_{{\cH^2}} \leq C(t) \norma{\hini}_{{\cH^2}}$\,.
\end{enumerate}
Using (i), (ii) and (iii) we have that $\Pi_N R^\eps(s) \Pi_N h^\eps(s) \in \cD(H_0)$
and, therefore, $S(t-s) \Pi_N R^\eps(s) \Pi_N h_\EM^\eps(s)$ is continuously 
differentiable in $s$.
This makes possible to perform an integration by parts in the integral defining $I_N(t)$.
Since
$$
  R^\eps(t) = \eps^2 \int R^\eps_\omega(t)\,dt,
$$
where
$$
  \left( R^\eps_\omega(t)h\right)_n(x) = 
  \sum_{n' \not= n} \frac{1}{i\omega_{nn'}}\,\e^{i\omega_{nn'}t/\eps^2} V_{nn'}(x)\, h_{n'}(x),
$$
then the integration by parts yields
\begin{multline*}
  I_N(t) = \eps^2 S(t-s)\Pi_N R^\eps_\omega(s)\Pi_Nh_\EM^\eps(s)\Big|^{s=t}_{s=0}
\\
  - \eps^2 \int_0^t S(t-s)\Pi_N
  \left[ iH_0R^\eps_\omega(s)\Pi_Nh_\EM^\eps(s) + R^\eps_\omega(s)\Pi_N\frac{d}{ds}h_\EM^\eps(s)\right]ds,
\end{multline*}
where, of course,
$$
 \frac{d}{ds}h_\EM^\eps(s) = H_0h_\EM^\eps(s) + R^\eps(s)h_\EM^\eps(s).
$$
Since $\Pi_N R^\eps_\omega(t)$ is uniformly bounded by some constant dependent of $N$ 
(in particular, such constant will depend of $1/\min\{ \omega_{nn'} \mid n'\not= n,\ n \leq N\}$), 
then, from (i), (ii) and (iii), 
and using $\Pi_N H_0 = H_0\,\Pi_N$, we obtain that a constant $C_N(\tau)$, 
independent of $\eps$, exists such that
$$
 \norma{I_N(t)}_{\cL^2} \leq \eps^2 C_N(\tau) \norma{\hini}_{{\cH^2}},
 \qquad 0 \leq t \leq \tau.
$$
Thus, fixing $N \geq N(\delta,\tau)$, a $\eps$ small enough exists 
such that $ \norma{I_N(t)}_{\cL^2} \leq \delta$, for all $0 \leq t \leq \tau$.
For such $N$ and $\eps$ we have, therefore,
$$
  \norma{h_\EM^\eps(t) - h_\EM(t)}_{\cL^2} \leq  \norma{I_N(t)}_{\cL^2} +  \norma{I_N^c (t)}_{\cL^2}
  \leq 2\delta,
$$
which proves the theorem in the regular case.
\paragraph{Case of general data}
If $\mu\geq 2$, then there is nothing to do. 
Let us assume $0< \mu < 2$ and let $\delta$ be a regularizing parameter and let
$\hini_\delta$ and $V_\delta$ be two regularizations of $\hini$ and of $V$ such that
$$
  \hini_\delta \in \cH^2,  \qquad 
  \lim_{\delta \to 0} \norma{\hini_\delta - \hini}_{\cH^\mu} = 0
$$
and
$$
 V_\delta \in \cW_2, \qquad 
 \lim_{\delta \to 0} \norma{V_\delta - V}_{\cW_\mu} = 0.
$$
Let $h_{\EM,\delta}^\eps$ and $h_{\EM,\delta}$ be the corresponding solutions of \eqref{hdyn} 
and \eqref{limit} with the modified initial data and potential.
Then we have
\begin{multline*}
  \norma{(h_\EM^\eps - h_\EM)(t)}_{\cL^2} 
  \leq \norma{(h_\EM^\eps - h_{\EM,\delta}^\eps)(t) }_{\cL^2} +
\\
  + \norma{ (h_{\EM,\delta}^\eps - h_{\EM,\delta})(t)}_{\cL^2}
  + \norma{ (h_{\EM,\delta} - h_\EM)(t)}_{\cL^2}.
\end{multline*}
The above analysis of the regular case shows that for any fixed $\delta > 0$, the second term of the 
right hand side tends to zero as $\eps$ tends to zero. 
Thanks to Theorem \ref{T3}, it is easy to show that the third term of the right hand side tends to zero 
as $\delta$ tends to zero and that the first term of the right hand also tends to zero as $\delta$ tends 
to zero uniformly in $\eps$. 
\end{proof}
\subsection{Convergence of the density}
\label{EMWF}
In this section, we prove the convergence of the particle density towards the superposition of the 
envelope function densities. 
Namely, we have the following theorem.
\begin{theorem}
\label{MainTheorem}
Let the initial datum $\psini \in L^2(\mR^d)$ be such that its envelope functions $(f_n^{in,\eps})$ 
form a bounded sequence in $\cH^\mu$ which strongly converges in $\cL^2$ towards the initial 
datum $f^{in} = (f^{in}_n)$, and assume that there exists a positive $\mu$ such that $V \in \cW_\mu$.
Then for any given function $\theta \in L^1(\mR^d)$ such that $\widehat{\theta}\in L^1(\mR^d)$, 
the following convergence holds locally uniformly in time:
$$
\lim_{\eps \to 0} \int \abs{\psi^\eps(t,x)}^2\, \theta(x) \, dx = \sum_n \int \theta (x) \abs{h_{\EM,n}(t,x)}^2\, dx,
$$
where $\psi_\eps$ is the solution of  \eqref{SE1} and $h_\EM$ is the solution of \eqref{limit}.
\end{theorem}
\begin{proof}
let $h^\eps_n(t,x) = f^\eps_n(t,x) \e^{iE_n t/\eps^2}$  where $f_n^\eps$ are  the envelope functions of $\psi^\eps$
We deduce from the results of the above subsection, in particular from Theorem \ref{LastTeo}, that
$$
  \lim_{\eps\to 0}  \sum_n \norma{h_n^\eps(t) - h_{\EM,n}(t)}_{L^2(\mR^d)}^2 = 0.
$$
Let 
$$
  \tilde{\theta}^\eps = \cT_{1\over 3\eps}(\theta), \quad \widetilde{h}_n^\eps =   \cT_{1\over 3\eps} (h_n^\eps), \quad \widetilde{h}_{\EM,n}^\eps =  \cT_{1\over 3\eps} (h_{\EM, n}) 
$$
where the truncation operator  $\cT_\gamma$ has been defined in \eqref{truncation}.
Recalling that
$$
  \psi^\eps(t,x) = \abs{\cC}^{1/2} \sum_{n} h_{n}^\eps(t,x) \e^{-iE_nt/\eps^2} v_n^\eps(x)
$$
let us define 
$$
  \tilde{\psi}^\eps(t,x) = \abs{\cC}^{1/2} \sum_{n} \tilde{h}_{n}^\eps(t,x) \e^{-iE_nt/\eps^2} v_n^\eps(x).
$$
It is readily seen, in view of \eqref{truc-error}  that 
$$
 \norma{\psi^\eps(t) - \tilde{\psi}^\eps(t)}_{L^2}^2 
 = \sum_n \norma{h_n^\eps(t) - \tilde{h}_n^\eps(t)}_{L^2}^2 
 \leq C \eps^\mu \norma{h^\eps(t)}_{\cH^\mu}^2,
$$
where, by Lemma \ref{L3}, $ \norma{h^\eps(t)}_{\cH^\mu}$ remains bounded.
It is now clear that 
$$
  \abs{\int \theta(x) \abs{\psi^\eps(t,x)}^2\, dx - \int \theta(x) \abs{\tilde{\psi}^\eps(t,x)}^2\, dx}  
  \leq \abs{\norma{\psi^\eps(t)}_{L^2}^2 -  \norma{\tilde{\psi}^\eps(t)}_{L^2}^2} \norma{ \theta}_{L^\infty} 
$$
goes to 0 and, therefore, we can replace $\psi^\eps$ by $\tilde{\psi}^\eps$.
Now,
$$
  \abs{ \int \theta(x) \abs{\tilde{\psi}^\eps(t,x)}^2\, dx -  \int \tilde{\theta}^\eps(x) \abs{\tilde{\psi}^\eps(t,x)}^2\, dx}
  \leq  \norma{\tilde{\psi}^\eps(t)}_{L^2}^2\norma{\tilde{\theta}^\eps-  \theta}_{L^\infty} \to 0
$$
and, therefore, we can replace $\theta$ by $\tilde{\theta}^\eps$.
But
\begin{multline*}
 \int \tilde{\theta}^\eps(x) \abs{\tilde{\psi}^\eps(t,x)}^2\, dx = 
 \\
  \abs{\cC} \left\langle \sum_{n}\tilde{\theta}^\eps (x)\tilde{h}_{n}^\eps(t,x) \e^{-iE_nt/\eps^2} 
 v_n^\eps(x)
 \, , \,
 \sum_{n} \tilde{h}_{n}^\eps(t,x) \e^{-iE_nt/\eps^2} v_n^\eps(x) \right\rangle
\end{multline*}
and $\supp(\widehat{\tilde{\theta}^\eps \tilde{h}_{n}^\eps}) \subset \cB/3\eps + \cB/ 3 \eps \subset  \cB/\eps$. 
Therefore the Parseval formula \eqref{Parsevaleps} shows that 
$$
  \int \tilde{\theta}^\eps(x) \abs{\tilde{\psi}^\eps(t,x)}^2\, dx = 
  \sum_n  \int \widetilde{\theta}^\eps (x) \abs{\widetilde{h}_n^\eps(t,x)}^2\, dx 
  \to \sum_n \int \theta (x) \abs{h_{\EM,n}(t,x)}^2\, dx,
$$
which completes the proof of the theorem.
\end{proof}
\section{Comments}
\label{sec6}
One of the most restrictive hypotheses that we made in the previous sections is the simplicity of all the 
eigenvalues of the periodic operator $\cH^1_\cL$.  
The question of simplicity of the eigenvalues is central in this problem as already has been noticed in 
the works of Poupaud and Ringhofer \cite{PoupaudRinghofer96} and of Allaire and Piatnistki \cite{Allaire05}. 
In these references,  the authors do not assume that all the eigenvalues are simple but assume that the 
initial datum is concentrated on finite number of bands who have multiplicity 1. 
The difference between our approach and that of these two references is that ours allows for a an 
infinite number of envelope functions. 
Besides, the hypothesis of simplicity of all the eigenvalues at $k = 0$ can be removed and replaced by the 
fact that the initial datum envelope functions corresponding to multiple eigenvalues are vanishing.
The proof has however to be reshuffled and we have chosen to stick to the restrictive hypothesis of 
simple eigenvalues. 
Let us however briefly explain how we can deal with this problem.
One important step is the diagonalization of the k$\cdot$p Hamiltonian which gives rise to the equation \eqref{GD}. 
In this formula the operator $\Lambda^\eps$ is diagonal in the $n$ index while $T^*_\eps\cU^0 T_\eps$ is not 
(the existence of the unitary transformation is still valid even in the case of multiple eigenvalues; it is continuous, 
but not regular for eigenvalues with multiplicity larger than one). 
Because of the separation of the eigenvalues, it is  easy to show that the eigenspaces with different energies 
are decoupled from each other (adiabatic decoupling) and we can replace $T^*_\eps\cU^0 T_\eps$
 by $\cU^0_{nn} \delta_{nn'}$. 
 If the initial data are only concentrated on modes with multiplicity one, then the solution itself is almost 
 concentrated on these modes and for these modes, we can make the expansion of eigenvalues and 
 obtain the effective mass equation \eqref{limit}. 
 Let us also mention a recent work by F. Fendt-Delebecque and F. M\'ehats  \cite{fanny} where the 
 effective mass approximation is performed for the Schr\"odinger equation with large magnetic field 
 and which relies on large time averaging of almost periodic functions. 
 This approach might be of help for analyzing the limit for multiple eigenvalues.
\par
\medskip
One final question which has not been addressed so far is the relationship between the regularity of function $\psi$ 
and that of its corresponding sequence $f^\eps$ of envelope functions. 
In particular, one may look for sufficient conditions on $\psi$ so that $f^\eps \in \cH^\mu$. 
Since the envelope function is a Fourier like expansion of the function $\psi$ on the basis $v_n$, then their 
decay as $n$  becomes bigger depends not only on the regularity of $\psi$ but also on that of the basis $(v_n)$ 
which itself will depend on the regularity of the potential $W_\cL$. 
We show in the following subsection some results in this direction.
\subsection{Asymptotic behavior of scaled envelope functions}
In this section we study the asymptotic behavior as $\eps$ tends to zero of the scaled envelope functions 
relative to the basis $(v_n)$ defined in \eqref{periodic}.
\par
From \eqref{EFChiEps}, it is readily seen that the limit as $\eps$ tends to zero of the envelope function 
is given by 
$$
\begin{array}{lll}
 \ds \lim_{\eps\to 0}  \hat f_n^\eps(k) &= &\ds \lim_{\eps \to 0}\abs{\cB\,}^{-1/2} \int_{\mR^d} \,\cara_{\cB/\eps}(k) 
 \,\e^{-ik\cdot x}\,v_n\left({x\over \eps}\right)\,\psi(x)\,dx
 \\[10pt]
 & = & \ds \abs{\cB\,}^{-1/2}\abs{\cC}^{-1}  \bk{v_n,1} \int_{\mR^d} \,\e^{-ik\cdot x}\psi(x)\,dx 
 =  \abs{\cC}^{-1/2} \bk{v_n,1}\, \widehat{\psi}(k).
 \end{array}
$$
Therefore
\begin{equation}
  \lim_{\eps\to 0} \pi_n^\eps(\psi) =   \abs{\cC}^{-1/2} \bk{v_n,1}\, \psi.
\end{equation}
The following Proposition, shows that the regularity of the crystal potential leads to decay properties 
on the coefficients $\bk{v_n,1} = \int_\cC v_n(x)\,dx$.
\begin{proposition}
\label{vn1}
Let $W_\cL$ be in $C^\infty$. 
Then for any integer  $p$, the coefficients $\bk{v_n,1}$  satisfy the inequality
$$
  \abs{\bk{v_n,1}} \leq {C_p \over E_n^p},
$$
where $C_p$ is a constant only depending on $\|W_\cL\|_{W^{2p, \infty}}$.
\end{proposition}
\begin{proof}
We first remark that
$$
  E_n^p \bk{v_n, 1}  = \bk{H_\cL^p v_n, 1} =  \bk{v_n, H_\cL^p  1}
$$
(where $H_\cL^p$ denotes the $p$-th power of $H_\cL$, not to be confused with the notation $H_\cL^\eps$
introduced in Sec.~\ref{sec2}).
Now it is readily seen that if $W_\cL \in W^{2p, \infty}$, then $H_\cL^p  1 \in L^\infty$ with
$\norma{H_\cL^p  1}_{L^\infty} \leq C \norma{W_\cL}_{W^{2p, \infty}}$, for a suitable constant $C\geq0$.
Then
$$
  E_n^p \abs{\bk{v_n, 1}} \leq \|v_n\|_{L^2} \norma{H_\cL^p  1}_{L^2} \leq C_p,
$$
with $C_p$ only depending on $\norma{W_\cL}_{W^{2p, \infty}}$, which ends the proof.
\end{proof} 
We also have the following property.
\begin{lemma}
\label{decay}
Let $\lambda$ and $\lambda'$ two elements of the reciprocal lattice $\cL^*$. 
Assume that $W_\cL \in C^\infty$. 
Then, for any integers $k, p$, we have the estimate
$$
  \abs{\bk{H_\cL^k \e^{i\lambda \cdot x}, H_\cL^k \e^{i\lambda'\cdot x} }} \leq 
  C_{k,p}{ (1+ \abs{\lambda}^{2k} \abs{\lambda'}^{2k}) \over  1 + \abs{\lambda - \lambda'}^{2p}},
$$
for a suitable constant $C_{k,p}\geq0$.
\end{lemma}
\begin{proof}
It is clear that 
$H_\cL^k \e^{i\lambda \cdot x} = \sum_{\abs{\alpha}=0}^{2k} \lambda^\alpha V_\alpha(x)  \e^{i\lambda \cdot x}$,
where $V_\alpha$ contains products of $W_\cL$ and its derivatives up to order $2k-\abs{\alpha}$.
Therefore
$$
  \bk{H_\cL^k \e^{i\lambda \cdot x}, H_\cL^k \e^{i\lambda'\cdot x} } = 
  \sum_{\abs{\alpha}, \abs{\beta} = 0}^{2 k}\lambda^\alpha (\lambda')^\beta 
  \int_\cC V_\alpha(x) V_\beta(x) \,\e^{i(\lambda-\lambda')\cdot x}dx.
$$
Now the result can be obtained by simply integrating by parts $2p$ times. 
\end{proof}
The estimate of Lemma \ref{decay} is not optimal and can certainly be refined,
but this is not the scope of our paper.
Next proposition follows from the previous result.
\begin{proposition}
Assume $W_\cL \in L^\infty$ and let $f_n^\eps = \pi_n^\eps(\psi)$ be the envelope functions
of $\psi$. 
Then the following estimate holds for any $\mu \geq 0$:
\begin{equation}
  \norma{f^\eps}_{\cL^2_\mu}^2 = 
  \sum_{n} \norma{ (1+ \abs{k}^2)^{\mu/2}\, \widehat{ f_n^\eps} (k)}_{L^2}^2 \leq C_\mu \|\psi\|_{H^\mu}^2.
\label{hshs}
\end{equation}
Let now $W_\cL$ be in $C^\infty$, then the following estimate holds for any integer $s$
\begin{equation}
\sum_{n} E_n^{s} \norma{f_n^\eps}_{L^2}^2 \leq C_s  (\|\psi\|_{L^2}^2 +\eps^{2s} \|\psi\|_{H^s}^2).
\label{enhs}
\end{equation}
\end{proposition}
\begin{proof}
Let us first prove \eqref{hshs}. 
Using the identity
$$
  \widehat{ f_n^\eps}(k) = \abs{\cB\,}^{-1/2} \int_{\mR^d} \,\cara_{\cB/\eps}(k) 
  \,\e^{-ik\cdot x}\,v_n\left({x\over \eps}\right)\,\psi(x)\,dx,
$$
as well as the decomposition
$$
 v_n (x) = {1\over \abs{\cC}^{1/2}} \sum_{\lambda \in \cL^*} v_{n,\lambda} e^{i\lambda\cdot x}
$$
where $v_{n,\lambda} = \bk{v_n, { e^{i\lambda\cdot x}\over \abs{\cC}^{1/2}} }$, 
we obtain, 
\begin{multline*}
 \sum_{n} \norma{ (1+ |k|^2)^{\mu/2}\,\widehat{ f_n^\eps}(k) }_{L^2}^2 
\\
 =  \sum_n \sum_{\lambda,\lambda'} \int_{\cB/\eps} 
  (1+|k|^2)^\mu  v_{n,\lambda}  \overline{v_{n,\lambda'}} \,
 \widehat{\psi}\left(k -\textstyle{{\lambda\over \eps}}\right) 
 \overline{\widehat{\psi}}\left(k -\textstyle{{\lambda'\over \eps}}\right) dk.
\end{multline*}
Summing first with respect to $n$ and using the identity 
$$
  \sum_nv_{n,\lambda}  \overline{v_{n,\lambda'}}  = 
  \frac{1}{\abs{\cC}} \,\langle e^{i\lambda\cdot x},  e^{i\lambda' \cdot x} \rangle
  = \delta_{\lambda,\lambda'},
$$
 the right hand side of the above identity takes the simple form
$$
  \ds \sum_{n} \norma{ (1+ |k|^2)^{\mu/2}\,\widehat{ f_n^\eps}(k) }_{L^2}^2 
  = \ds \sum_{\lambda \in \cL^*}  \int_{\cB/\eps} (1+|k|^2)^\mu  
  \abs{\widehat{\psi}\left(k -\textstyle{{\lambda\over \eps}}\right)}^2 dk.
$$
It is now readily seen that there exists a constant $c \geq 1$, only depending on the fundamental cell $\cC$, 
such that for all $k\in\cB$ and for all $\lambda \in \cL^*$, we have the estimate
$$ 
  \abs{k} \leq c \abs{k-\lambda},
$$ 
so that 
\begin{multline*}
  \sum_{n} \norma{ (1+ \abs{k}^2)^{\mu/2} \,\widehat{ f_n^\eps}(k)}_{L^2}^2  \leq
  c^{2\mu} \sum_{\lambda \in \cL^*}  \int_{\cB/\eps} \left(1+\abs{k-\textstyle{{\lambda \over \eps}}}^2\right)^\mu   
  \abs{\widehat{\psi}\left(k -\textstyle{{\lambda \over \eps}}\right)}^2\, dk  
  \\
  = c^{2\mu}  \int_{\mR^d} (1+\abs{k}^2)^\mu   \abs{\widehat{\psi}(k)}^2\, dk.
\end{multline*}
This implies that a suitable constant $C_\mu$ exists such that  \eqref{hshs} holds. 
Let us now prove \eqref{enhs}.
We proceed analogously and find
$$
  \sum_{n} E_n^{s} \norma{f_n^\eps }_{L^2}^2 
   = {1\over (2\pi)^d} \ds \sum_n\sum_{\lambda,\lambda'} \int_{\cB/\eps} E_n^{s} \, 
   v_{n,\lambda}  \overline{v_{n,\lambda'}} \,\widehat{\psi}\left(k -\textstyle{{\lambda\over \eps}}\right)
   \overline{\widehat{\psi}}\left(k -\textstyle{{\lambda'\over \eps}}\right) dk.
$$
As above, we first make the sum over the index $n$ and, therefore, we need to evaluate
$$
 \sum_{n} E_n^{s}  v_{n,\lambda}  \overline{v_{n,\lambda'}}.
$$
We first remark that $E_n^s  v_{n,\lambda}  =\bk{ H_{\cL}^s v_n, { e^{i\lambda\cdot x}\over \abs{\cC}^{1/2}} } 
= \bk{ v_n, H_{\cL}^s { e^{i\lambda\cdot x}\over \abs{\cC}^{1/2}} }$.
Therefore
$$
  \sum_{n} E_n^{s} \, v_{n,\lambda}  \overline{v_{n,\lambda'}} = 
  {1\over \abs{\cC} } \bk{ H_{\cL}^s   e^{i\lambda\cdot x},     e^{i\lambda'\cdot x} }.
$$
Contrary to the proof of \eqref{hshs},  the obtained formula is not diagonal in $(\lambda, \lambda')$ but 
Lemma \ref{decay} leads to the following estimate, which holds for large enough integers $p$:
$$
\begin{aligned}
 \sum_{n} E_n^{s} \norma{f_n^\eps}_{L^2}^2 
 &\leq C_{s,p} \sum_{\lambda,\lambda'} \int_{\cB/\eps} 
 {1 + \abs{\lambda}^{2s}  \over 1 + \abs{\lambda-\lambda'}^{2p}}  
\abs{\widehat{\psi}\left(k -\textstyle{{\lambda\over \eps}}\right) 
\overline{\widehat{\psi}}\left(k -\textstyle{{\lambda'\over \eps}}\right) } dk 
\\[6pt]
&\leq \frac{C_{s,p}}{2} \sum_{\lambda,\lambda'} \int_{\cB/\eps} 
{1 + \abs{\lambda}^{2s}  \over 1 + \abs{\lambda-\lambda'}^{2p}} 
 \left[\abs{\widehat{\psi}\left(k -\textstyle{{\lambda\over \eps}}\right)}^2 + 
 \abs{\overline{\widehat{\psi}}\left(k -\textstyle{{\lambda'\over \eps}}\right) }^2\right] dk 
\\[6pt]
&\leq C_{s,p} \sum_{\lambda\in\cL^*}  \int_{\cB/\eps} (1 + \abs{\lambda}^{2s} )  
\abs{\widehat{\psi}\left(k -\textstyle{{\lambda\over \eps}}\right)}^2\, dk.
\end{aligned}
$$
Note that we used the fact that, for large enough $p$, the following estimates hold with constants 
$C_1$ and $C_2$ only depending on $s$ and $p$
$$
  \sum_{\lambda \in \cL^*} {1 + \abs{\lambda}^{2s}  \over 1 + \abs{\lambda-\lambda'}^{2p}} 
  \leq C_1 (1+ \abs{\lambda'}^{2s}),  
  \quad  
  \sum_{\lambda' \in \cL^*} {1 + \abs{\lambda}^{2s}  \over 1 + \abs{\lambda-\lambda'}^{2p}} 
  \leq C_2 (1+ \abs{\lambda}^{2s}).
$$
Now, for $\lambda \neq 0$ and $\eps k \in \cB$  it is readily seen that $|\lambda| \leq c_0 \abs{\lambda -\eps k}$, 
where $c_0$ is a positive constant independent of $\lambda$ and $k$. 
Therefore,
$$
  \sum_{\lambda\in\cL^*}  \int_{\cB/\eps} (1 + \abs{\lambda}^{2s} )  
  \abs{\widehat{\psi}\left(k -\textstyle{{\lambda\over \eps}}\right)}^2\, dk 
  \leq \norma{\psi}_{L^2}^2 + \eps^{2s} c_0^{2s} \norma{\psi}_{H^s}^2,
$$
which implies that a suitable constant $C_s$ exists such that  \eqref{enhs} holds.
\end{proof}
\section{Postponed proofs}
\label{post}
This section is devoted to the proofs of some results stated in the beginning of the paper.
\subsection{Proof of Theorem \ref{T1}}
For any Schwartz function $\psi$ we can write
$$
  \psi(x) = (2\pi)^{-d/2} \int_{\mR^d} \hat \psi(k)\,\e^{ik\cdot x} dk
  = \sum_{\eta \in \cL^*} (2\pi)^{-d/2} 
  \int_{\cB + \eta} \hat \psi(k)\,\e^{ik\cdot x} dk
$$ $$
  = \sum_{\eta \in \cL^*} (2\pi)^{-d/2}\, \e^{i\eta\cdot x}
  \int_{\cB} \hat \psi(\xi+\eta)\,\e^{i\xi\cdot x}\, d\xi
  = \sum_{\eta \in \cL^*} \e^{i\eta\cdot x} G_\eta(x), 
$$
where
$$
  G_\eta(x) =  (2\pi)^{-d/2} \int_{\cB} \hat \psi(\xi+\eta)\,\e^{i\xi\cdot x}\,d\xi
$$
clearly belongs to $\cF^*L^2_\cB(\mR^d)$.
Moreover, we have
$$
 \sum_{\eta \in \cL^*} \norma{G_\eta}_{L^2}^2 = 
 \sum_{\eta \in \cL^*} \norma{\hat G_\eta}_{L^2}^2 =
 \sum_{\eta \in \cL^*} \int_{\mR^d} \abs{\hat \psi(\xi+\eta)\cara_\cB(\xi)}^2 dk 
$$ $$
 = \sum_{\eta \in \cL^*} \int_{\cB+\eta} \abs{\hat \psi(\xi)}^2 dk 
 = \norma{\psi}_{L^2}^2.
$$
Thus, defining
\begin{equation}
\label{Faux}
  F(x,y) = \sum_{\eta \in \cL^*} \e^{i\eta\cdot x} G_\eta(y), 
  \qquad
  (x,y) \in \cC\times\mR^d,
\end{equation}
we have that $F \in L^2(\cC\times\mR^d)$ and 
$$
  \abs{\cC}^{-1}\norma{F}_{L^2(\cC\times\mR^d)}^2 
  = \sum_{\eta \in \cL^*} \norma{G_\eta}_{L^2(\mR^d)}^2
  =  \norma{\psi}_{L^2(\mR^d)}^2 
$$
(where we used the fact that $\{ \abs{\cC}^{-1/2}\,\e^{i\eta\cdot x} \mid \eta \in \cL^* \}$ 
is a orthonormal basis of $L^2(\cC)$).
Since $\{ v_n \mid n \in \mN\}$ is another orthonormal basis of $L^2(\cC)$, then we can also write
$$
 F(x,y) = \abs{\cC}^{1/2}\,\sum_n f_n(y) v_n(x),
$$
where
\begin{equation}
 f_n(y) = \abs{\cC}^{-1/2} \bk{F(\cdot,y),v_n}_{L^2(\cC)}.
\end{equation}
Note that $\hat f_n \in L^2_\cB(\mR^d)$ for every $n$ and that  
$$
  \norma{\psi}_{L^2(\mR^d)}^2 = \abs{\cC}^{-1} \norma{F}_{L^2(\cC\times\mR^d)}^2  
  = \sum_n\norma{f_n}_{L^2(\mR^d)}^2.
$$
For $y=x$, \eqref{Faux} yields \eqref{EFdec}, at least for Schwartz functions.
However, it can be easily proved that the mapping $\psi \mapsto (f_0,f_1,\ldots)$
can be uniquely extended to an isometry between $L^2(\mR^d)$ 
and $\ell^2(\mN,\cF^*L^2_\cB(\mR^d))$, with the properties \eqref{EFdec} and \eqref{Parseval}.
\subsection{Proof of Theorem \ref{T2}}
Recalling definition \eqref{truncation}, let
$$
  \tilde{\theta}^\eps = \cT_{1\over 3\eps}(\theta), \qquad \tilde{f}_n^\eps =   \cT_{1\over 3\eps} (f_n^\eps), 
$$
and define
$$
  \tilde{\psi}^\eps(x) = \abs{\cC}^{1/2} \sum_n \tilde{f}_n^\eps(x)\, v_n^\eps(x).
$$
Then, we can write
\begin{equation*}
\begin{aligned}
  &\int_{\mR^d} \theta(x)\Big[ |\psi(x)|^2 -\sum_{n} |f_n^\eps(x)|^2\Big] dx 
  = \int_{\mR^d} \theta(x)\Big[\abs{\psi(x)}^2 -|\tilde{\psi}^\eps(x)|^2\Big]dx  
\\[6pt]
 &+ \int_{\mR^d} \Big[ \theta(x)- \tilde{\theta}^\eps(x) \Big] |\tilde{\psi}^\eps(x)|^2\, dx  
 + \int_{\mR^d} \tilde{\theta}^\eps(x)\Big[|\tilde{\psi}^\eps(x)|^2 -\sum_{n} |\tilde{f}_n^\eps(x)|^2\Big] dx  
\\[6pt]
 &+ \int_{\mR^d} \tilde{\theta}^\eps(x) \sum_n \Big[|\tilde{f}^\eps_n(x)|^2 - |f_n^\eps(x)|^2\Big] dx  
 + \int_{\mR^d} \Big[ \tilde{\theta}^\eps(x)- \theta(x) \Big] \sum_n |f^\eps_n(x)|^2\, dx  
\\[6pt]
 & = I_1 + I_2 + I_3 + I_4 + I_5.
\end{aligned}
\end{equation*} 
Since $\supp (\widehat{\tilde{\theta}^\eps \tilde{f}^\eps_n }) \subset \cB/3\eps + \cB/3\eps \subset \cB/\eps$,
then $\tilde{\theta}^\eps \tilde{f}^\eps_n$ are the envelope functions of $\tilde{\theta}^\eps \tilde{\psi}^\eps$ 
and the Parseval identity \eqref{Parsevaleps} can be applied to the functions  $\tilde{\psi}^\eps$ and 
$\tilde{\theta}^\eps \tilde{\psi}^\eps$, which yields $I_3 = 0$.
\par
As far as the terms $I_2$ and $I_5$ are concerned, we have
$$
  \abs{I_2} \leq \norma{\theta- \tilde{\theta}^\eps}_{L^\infty} \norma{\tilde{\psi}^\eps}_{L^2}
  = \norma{\theta- \tilde{\theta}^\eps}_{L^\infty} \sum_n \norma{\tilde{f}^\eps_n}_{L^2}
  \leq \norma{\theta- \tilde{\theta}^\eps}_{L^1} \norma{\psi}_{L^2}
$$
and, therefore, $I_2 \to 0$ as $\eps \to 0$. 
Similarly we can prove that $I_5 \to 0$.
\par
Finally, if $R$ is the radius of a ball contained in $\cB$, we have
\begin{multline*}
  \abs{I_1}  \leq  
  \norma{\theta}_{L^\infty}  \Big( \norma{\psi}_{L^2}^2 - \norma{\tilde{\psi}^\eps}_{L^2}^2 \Big) =
  \\
  \norma{\theta}_{L^\infty} \sum_n \Big( \norma{f_n^\eps}_{L^2}^2 - \norma{\tilde{f}^\eps_n }_{L^2}^2 \Big)
 \leq \norma{\theta}_{L^\infty} \sum_n  \int_{\abs{k} > \frac{R}{3\eps}} \abs{\hat {f}^\eps_n(k)}^2 dk.
\end{multline*}
The last integral goes to 0 as $\eps\to 0$, because 
$\sum_n  \int_{\mR^d} |\hat {f}^\eps_n(k)|^2 dk = \norma{\psi}_{L^2}^2$ and the dominated convergence 
theorem applies.
Thus $I_1 \to 0$ and, in a similar way, we can also prove that $I_4\to 0$. 
In conclusion, 
$$
  \int_{\mR^d} \theta(x)\Big[ |\psi(x)|^2 -\sum_{n} |f_n^\eps(x)|^2\Big] dx \to 0
$$
as $\eps \to 0$, which proves the theorem.
\medskip 

{\textbf Acknowledgements.}
N. Ben Abdallah acknowledges support from the project 
QUATRAIN (BLAN07-2 212988) funded by the French Agence Nationale de la Recherche) and from the 
Marie Curie Project DEASE: MEST-CT-2005-021122 funded by the European Union.
L. Barletti acknowledges support from Italian national research project PRIN 2006 
``Mathematical modelling of semiconductor devices, mathematical methods
in kinetic theories and applications'' (2006012132\_004).

\end{document}